\documentclass[a4paper,10pt]{article}
\usepackage{graphicx}
\usepackage{amsmath,amssymb,amsfonts,amsthm}
\usepackage{url, enumerate,anysize}
\usepackage{lscape}
\usepackage[subnum]{cases}

\usepackage[usenames,dvipsnames,svgnames,table]{xcolor}
\usepackage{color}
\usepackage[colorlinks=true,linkcolor=black, citecolor=blue, urlcolor=blue]{hyperref}
\usepackage{graphicx}
\usepackage{tikz}
\usepackage[titletoc]{appendix}
\newtheorem{thm}{Theorem}[section]

\newtheorem{lem}[thm]{Lemma}
\newtheorem{defn}[thm]{Definition}
\newtheorem{rmk}[thm]{Remark}
\newtheorem{prop}{Proposition}

\usepackage[usenames,dvipsnames,svgnames,table]{xcolor}

\title{Load balancing with heterogeneous schedulers}
\author{Urtzi Ayesta, Manu K. Gupta\footnote{Corresponding author: Manu K. Gupta (manu-kumar.gupta@irit.fr)}, and Ina Maria Verloop\\IRIT, 2 rue C. Camichel, Toulouse, France}

\marginsize{1.5cm}{1.5cm}{1.5cm}{1.5cm}

\begin{document}
\maketitle

\begin{abstract}
Load balancing is a common approach in web server farms or inventory routing problems. An important issue in such systems is to determine the server to which an incoming request should be routed to optimize a given performance criteria. In this paper, we assume the server's scheduling disciplines to be heterogeneous.  More precisely, a server implements a scheduling discipline which belongs to the class of limited processor sharing (LPS-$d$) scheduling disciplines. Under LPS-$d$, up to $d$ jobs can  be served simultaneously, and hence, includes as special cases First Come First Served ($d=1$) and Processor Sharing ($d=\infty$).

In order to obtain efficient heuristics, we model the above load-balancing framework as a multi-armed restless bandit problem. Using the relaxation technique, as first developed in the seminal work of Whittle, we derive Whittle's index policy for general cost functions and obtain a closed-form expression for Whittle's index in terms of the steady-state distribution.  
Through numerical computations, we investigate the performance of Whittle's index with two different 
performance criteria: linear cost criterion and a cost criterion that depends on the first and second moment of the 
throughput. Our results show that \emph{(i)} the structure of Whittle's index policy can strongly depend on the scheduling discipline implemented in the server, i.e., on $d$, and that \emph{(ii)} Whittle's index policy significantly outperforms standard dispatching rules such as 
Join the Shortest Queue (JSQ), Join the Shortest Expected Workload (JSEW), and Random Server allocation (RSA).
\end{abstract}

\textbf{keywords:} Queuing, load balancing, restless bandits, limited processor sharing, index policies.

\section{Introduction}
We address the problem of load balancing in a multi-server system with heterogeneous service disciplines. In such systems, a dispatcher seeks to balance the assignment of service requests (jobs) across many servers in the system in order to optimize the performance. Such models are ubiquitous in applications, for instance in clusters of web-server nodes, database systems, grid computing and inventory routing (see for example \cite{glazebrook2009index, gupta2007analysis, zhang2008web}). The literature on load-balancing is vast, with many papers published every year (see Section~\ref{related} for a brief summary), but a common assumption in most of the literature is that all the servers implement the same scheduling policy. In particular, the overwhelming majority of works assume that servers implement FCFS, even though Processor-Sharing (PS) has also received some attention. This is the motivation for our work, where we aim at studying, to the best of our knowledge for the first time, how a job dispatcher should operate in the presence of heterogeneous servers that might vary in terms of speed as well as in terms of scheduling discipline.

In order to do so, we formulate the above load-balancing problem as a discrete Markov decision process (MDP). 
New jobs arrive to the system according to a Bernoulli process of mean $p$, and upon arrival of a new job, the dispatcher  must decide to which server to dispatch the new job, with the objective of optimizing the long-run performance. 

Within a server, we assume that jobs are served according to the Limited Processor Sharing (LPS) discipline. In LPS-$d$, $d$ jobs (if present) are served simultaneously, that is, the server equally distributes its  attention to each of these $d$ jobs. First Come First Served (FCFS) and Processor Sharing (PS) disciplines are special cases of LPS-$d$ with $d=1$ and $d=\infty$, respectively. Thus, every server is characterized by the scheduling discipline it deploys, i.e., the  value of $d \in \mathbb{N}$, and its capacity is measured by $q \in (0,1]$, the averaged number of jobs that he can serve in a time unit. The evolution of every server is independent of each other, and their behavior is  coupled only through the policy that dispatches jobs among them. This allows us to cast the problem within the so-called Multi-Armed Restless Bandit problems (MARB), where a bandit models a server, a class of MDP's that has been widely studied in the literature for its multiple applications \cite{gittins2011multi}. The multi-armed restless bandit problem has been shown to be \textit{PSPACE-complete} \cite{papadimitriou1999complexity}, and closed-form solutions can only be obtained in toy examples.

We thus consider the relaxed version of the problem, as pioneered by Whittle in \cite{whittle1988restless}, in which \emph{on average} $p$ jobs are dispatched. Under the technical condition known as \emph{indexability}, Whittle showed that the solution to the relaxed problem is of index type: for every server there exists a function, the index,  that depends only on its own state, and the optimal policy dispatches to the servers with highest current index. Whittle then defined a heuristic for the original problem, referred to as Whittle's index policy, where in every decision epoch the bandit with highest Whittle index is selected. It has been shown that the Whittle's index
policy performs strikingly well, see \cite{nino2007dynamic} for a discussion, and
is asymptotically optimal under certain conditions, see \cite{verloop2016asymptotically} and \cite{weber1990index}.

Unfortunately, Whittle's index can not always be calculated in closed form. For instance, Borkar et al. \cite{borkar2017whittle} consider a similar problem to ours with $d=\infty$, and Whittle's index is calculated via an iterative scheme. In our main methodological contribution, we provide a closed-form expression for Whittle's index for a LPS-$d$ server,  which allows us to investigate the performance and properties of Whittle's index policy  as a function of server's heterogeneity, that is, the values of $d$ and $q$. In order to do so we consider two specific cost functions: \emph{(i)} a linear cost function, and \emph{(ii)} a cost criterion that depends on the first and second moment of the  throughput. A linear cost function has often been used as a cost criteria in the literature in the context of computer services (see \cite{mendelson1985pricing} for a discussion), whereas the weighted sum of linear cost and second order throughput captures the mean-variance trade-off (see \cite{singh2015index}).  Under the linear cost criterion, we observe numerically that $q$ is the key parameter characterizing the structure of Whittle's index, while the impact of $d$ on the dispatching decision is negligible.  Under the mean-variance trade-off cost criterion, we observe that Whittle's index policy strongly depends on the value of $d$. %In both cases, we conclude that Whittle's index policy is close to the optimal performance across all the load factors, and that it  outperforms well-known dispatching rules such as Join the Shortest Queue (JSQ), Join the Shortest Expected Workload (JSEW), and Random Server allocation (RSA).
 In both cases, Whittle's index policy depends on the arrival probability $p$, which is key in making Whittle's index policy more efficient than dispatching rules that are oblivious to $p$, such as Join the Shortest Queue (JSQ), Join the Shortest Expected Workload (JSEW), and Random Server allocation (RSA). Our numerical results also indicate that Whittle's index policy is close-to-optimal under a wide range of parameters.

%We indeed achieve this and derive a closed form expression for indices in terms of stationary probabilities for general cost functions (see Proposition \ref{proposition:monotone_index} in Section \ref{sec:derivationWhittle}). Further, we simplify the expression of Whittle's index for FCFS schedulers ($d=1$). In the case of linear cost, an important insight is the accountability of \textit{future arrivals} in the index policy, whereas standard dispatching rules, such as JSQ or JSEW, fail to account for the same (see Section \ref{FCFS_linear_cost_index}). 
 
In summary, the main contributions of this paper are:
\begin{itemize}
\item We characterize in closed-form Whittle's index for a LPS-$d$ server and apply it to study the problem of dispatching jobs to heterogeneous servers.

\item Investigate the performance of Whittle's index policy as a function of the scheduling discipline of servers parameterized by $d$, the speed as measured by $q$, and the arrival probability given by $p$. 

\item Conclude that the impact of $d$ and $q$ vary depending on the objective. For  linear cost criterion, the speed $q$ is the main parameter while  $d$  does not have a big impact on the structure of Whittle's index policy. However, for a cost criterion that depends on the first and second moment, the dispatching policy strongly depends on both $d$ and $q$. 

\item Observe that Whittle's index policy can outperform known dispatching policies like JSQ (oblivious to $d$, $q$, and $p$) and JSEW (oblivious to $d$ and $p$). 
\end{itemize}

The remainder of the paper is organized as follows. Section~\ref{related} presents related work. Section \ref{sys_des} presents the load-balancing problem and explains our solution approach from a restless bandit perspective. Section \ref{sec:derivationWhittle} is dedicated to develop a closed-form expression for Whittle's index in terms of stationary probabilities. Section \ref{indices_for_different_discipline}  simplifies the expression in the case of  FCFS schedulers and provides insights on the index for the linear cost criterion. In Section \ref{scheduling_schemes}, we discuss the details of numerically finding an optimal policy and dispatching rules. We present our extensive numerical findings in Section \ref{numerics} for two different cost criteria. A concluding remark can be found in  Section \ref{sec:conclusion}.

\section{Related literature}
\label{related}
In this section we provide an overview of the most important related work  to the problem under consideration.

As mentioned in the introduction, literature on load balancing is vast and keeps growing motivated by application domains like cloud computing and parallel computing systems. 
A textbook covering load balancing is \cite[Chapter 24]{HB13}, and a  survey covering recent developments is \cite{BorstSurvey2018}. The large majority of papers consider FCFS implemented in servers, and a few works consider PS. As mentioned in the introduction, we are not aware of any work that analyzes how to dispatch jobs when servers deploy different scheduling disciplines. All the works mentioned in this overview assume FCFS, unless otherwise stated. A classical result shows that in case the dispatcher has precise information on the queue length in each  server,  and the servers are homogeneous, JSQ minimizes the total number of jobs in the system, see \cite{winston1977}. 
{Similarly, when the expected workload in each server is known to the dispatcher, JSEW has been shown to be optimal, see for example \cite{liu1998optimal, moyal2017pathwise}.} 
%From a practical perspective, it is not realistic to assume that dispatchers have this information, and this has triggered an enormous effort form the research community to analyze dispatching policies that do not rely on precise full-state information.  
A very well-known policy is the so-called Power-of-$d$ policy, according to which upon the arrival of a new job, the dispatcher will probe $d$ servers, and dispatch the job to the shortest queue among these $d$ servers, see \cite{Mitz01}. In another stream of work, researchers have investigated the performance of systems in which idle servers send a notification back to the dispatcher, see \cite{fosssto17}. In size-based dispatching, the dispatching decision is based on the size of the incoming job, see for example \cite{Harhol-Vesilo10,FMD05}. With the so-called pull/push mechanism, servers that are idle can claim jobs from busy servers, and vice versa, servers with long queues can transfer jobs to idle servers, see \cite{Push17} for an analysis. In recent years, researchers have also investigated redundancy models, in which multiple copies of the same job are dispatched to a subset of servers, see \cite{Gardner17}. The optimal size-based dispatching policy with PS servers is characterized in \cite{AAP09}.

%Routing policies such as JSQ, JSEW and static schemes have been analyzed for load balancing problem (see \cite{gupta2007analysis, zhang2012analysis}). JSEW has been shown to be optimal under a range of conditions. See, for example \cite{hordijk1990optimality, johri1989optimality, weber1978optimal}. However, Whitt \cite{whitt1986deciding} showed that JSEW might not be optimal for very simple cases. 

Restless bandits are a popular framework in various application domains, including inventory routing \cite{archibald2009indexability}, machine maintenance \cite{glazebrook2005index}, cloud computing \cite{borkar2017index}, sensor scheduling \cite{nino2011sensor}, etc. We refer to \cite{glazebrook2014stochastic} for a recent survey on the application of index policies in scheduling. Book length treatments of restless bandits can be found in \cite{jacko2010dynamic} and \cite{ruiz2008indexable}. 
%(see  \cite{borkar2017whittle, glazebrook2002index, glazebrook2009index, TON_maialen} etc.). 
%The Whittle index heuristic is prominent in various resource allocation problems, 
In the particular area of load balancing, there has been several previous works that have applied the restless banding framework, and who have calculated Whittle's index  either in closed-form  or via an iterative scheme (see \cite{argon2009dynamic, borkar2017whittle, glazebrook2004index, nino2002dynamic, nino2012towards}). For example, in \cite{borkar2017whittle}, servers are PS and an iterative scheme is reported to compute Whittle's index in the case of linear cost functions and no blocking. The load balancing problem in \cite{argon2009dynamic, glazebrook2004index} deals with FCFS servers with dedicated arrivals for each queue. Dedicated arrivals  have priority over new arrivals in the load balancing model of \cite{glazebrook2004index}. In \cite[section 8.1]{nino2002dynamic} the author considers finite buffers and FCFS servers and determines an index policy. The model in \cite{nino2012towards} considers instead the  objective of minimizing the average job loss rate. 
%We generalize the previous work by finding a closed form expression of Whittle's index in a general scheduling framework (LPS-$d$ disciplines) with general cost functions. 

An important modeling assumption we make is that time is discrete. In addition to the obvious mathematical tractability, from a practical point of view, it can be argued that decisions are often made in regular time moments instead of continuously.  Example application areas include web servers, distributed caching systems, large data stores and grid computing (see \cite{foster2008cloud, gupta2007analysis}). 
Discrete-time queues have been used to investigate the behavior of communication and computer systems in which time is slotted (see  \cite{artalejo2003performance, bruneel1994analysis, bruneel1993performance, laevens1998discrete}). For a comprehensive treatment of discrete time queues, we refer to recent books (see \cite{bruneel2012discrete, hunter2014mathematical}). The load balancing problem has been previously studied in discrete time for distributed systems in the presence of time delays (see \cite{dhakal2003dynamical}). 
% Additionally, discrete-time queuing models have also been used to describe queuing phenomena in digital computer and communication systems and networks, in which digital information is exchanged in the form of fixed-length packets, each requiring a fixed-length transmission time, which can then be taken as a discrete time-unit in the model.

%Most of the prior work in load-balancing problems assumes that all servers have implemented the same service discipline, see for example~\cite{altman2011load, borkar2017whittle, liu1998optimal, mukhopadhyay2013analysis}.

%The problem of load-balancing has been explored by many researchers (see \cite{altman2011load}, \cite{liu1998optimal}, \cite{mukhopadhyay2013analysis} and references therein). Most of the prior work in load-balancing assumes that all servers have implemented the same service discipline.Load balancing when having either only First come first served (FCFS) servers or only processor sharing (PS) servers  are  considered in the literature  (see for example \cite{lu2011join}, \cite{mukhopadhyay2013analysis}, \cite{weber1978optimal}). 

The single server under the LPS-$d$ policy has been widely studied for the particular cases of $d=1$ (FCFS) and $d=\infty $ (PS). For $ 1< d <\infty$, analysis is scarce due to its complexity. Avi-Itzhak $\&$ Halfin \cite{LPSd1} propose an approximation for the mean response time assuming Poisson arrivals. A computational analysis based on matrix geometric methods is developed by Zhang $\&$ Lipsky \cite{zhang2007analytical, zhang2006modelling}. Some stochastic ordering results are derived in \cite{nuyens2009monotonicity}. Zhang, Dai $\&$ Zwart \cite{zhang2009law, zhang2011diffusion, zhang2008steady} develop fluid, diffusion and heavy traffic approximations. Gupta $\&$ Harchol-Balter \cite{gupta2009self} consider approximation methods and Markov decision techniques to determine the optimal level $d$ when the system is not work-conserving. The sojourn time tail asymptotics for the LPS-$d$ queue for both heavy-tailed and light-tailed job size distributions are recently explored in \cite{nair2010tail}. However, none of the prior work has focused on load balancing in a system with  LPS-$d$ servers, which is the central theme of this work.

\section{Model description}\label{sys_des}
 We consider a slotted-time model with decision epochs $t \in  \mathcal{T}:=\{0, 1, 2, . . . \}$. Time epoch $t$ corresponds to the beginning of time period $t$. 
Jobs arrive according to a Bernoulli arrival process  with arrival probability $p\in(0,1)$.
 New arrival must be dispatched to at most one of the $K$ servers, or must be blocked (see Figure \ref{fig:load_balancing}). Server~$k$ serves the jobs in his queue according to the LPS-$d_k$ service discipline (defined below), where $d_k$ is a parameter determining the scheduling discipline. Since  $d_i$ can be different from $d_j$, $i,j=1,\ldots, K$, this models heterogeneous scheduling disciplines. We assume that the servers are independent, and have capacity $q_k$.

%Throughout this paper, we consider servers that evolve as discrete time Markov chains (DTMC), i.e, when server $k$ is in state $n_k$, it changes the state after one time unit, and the next state depends only on its current state ($n_k$). We assume that the dynamics of each server is independent of the others. We now explain the dynamics of different service disciplines. 

\begin{figure}[ht]
\centering
\resizebox{.48\textwidth}{!}{\input{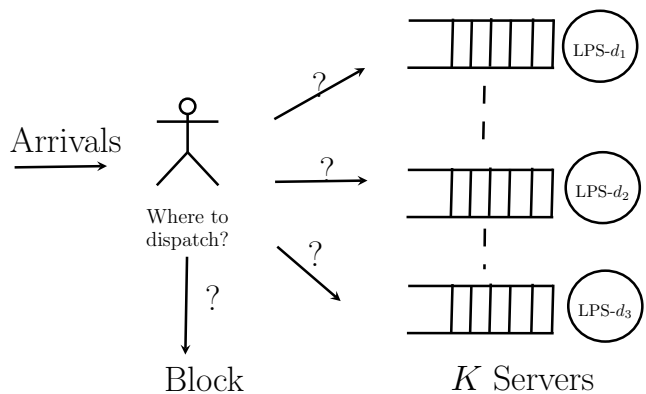}}
\caption{Abstraction of load balancing problem in a multi-server system with heterogeneous service disciplines.}\label{fig:load_balancing}
\end{figure}

Under the LPS-$d_k$ service discipline,  the first  $d_k$  jobs are served simultaneously and equally share the capacity of the server $k$. More precisely, the probability of departure in server~$k$ for the first $\min(d_k, n)$ jobs, where $n$ is the total number of jobs present in the queue of server $k$, is given by $q_k/\min(d_k, n)$.  
%If the total number of jobs ($n$), $n\le d$, LPS-$d$ systems are same as PS systems; whereas for $n>d$, LPS-$d$ systems work as follows. The entire processor capacity is shared between first $d$ jobs. 
The departure process from server~$k$ will be a binomial process, with the following properties. If there is only one job in the server, the probability of it being completed in a given time slot
is $q_k$. However, if there are $m$ jobs in the server, by the egalitarian scheme, the probability of departure will be reduced to $q_k/m$ for each of the (up to $d_k$) jobs in the server.  The number of departures from the server will be a binomial random variable with departure probability $q_k/m$,
and the the mean number of departures in a time slot remains fixed at $q_k$. The departure process from server $k$ can jump down up to state $n-d_k$ in one step. The transition probability for $i$ departures in an LPS-$d_k$ system when there are in total $n$ jobs, is given by,  
\begin{equation}\label{LPSD_prob}
 q_k^{d_k}(i|n):= {\min\{n,d_k\} \choose i} \left(\frac{q_k}{\min\{n,d_k\}}\right)^i\left(1-\frac{q_k}{\min\{n,d_k\}}\right)^{\min\{n,d_k\}-i},
\end{equation}
for $i \le \min\{n,d_k\}$, and is equal to 0 for $i > \min\{n,d_k\} $.

It can be easily seen that FCFS and PS   are special cases of the LPS-$d_k$ scheduling scheme with $d_k=1$ and $\infty$, respectively. 
Under FCFS, a job in front of the queue is always served with full capacity, hence $d_k=1$. If the capacity of server~$k$ is $q_k$, the probability that a job in front of the queue will be completed in a given time slot is $q_k$. Hence, in one unit of time, the  process describing the number of jobs in the server can jump down at most by one.
Under PS, all jobs in the system are given the same processing power, hence $d_k=\infty$. Thus, there is a strictly positive probability that the process describing the number of jobs in the server goes in one time unit to zero, that is, all $n$ jobs present in the server depart.

%\section{Restless bandit model}
%Consider the slotted time with decision epochs $t \in  \mathcal{T}:=\{0, 1, 2, . . . \}$. Time epoch $t$ corresponds to the beginning of time period $t$. 
%We consider a stochastic resource allocation problem with $K$ different servers. Let $N_k(t) \in \{0,1,...\}$ denote the state of server $k$ at decision epoch $t$, $k=1,...,K$. 
%At each decision epoch, the controller can choose for each server between two actions: action $a=0$, that is, arrival is blocked, or action $a=1$, that is, arrival is routed to the server, with the restriction that at any moment in time, arrival can be routed to at most one (out of $K$) server. 

A policy $\phi$ decides how new arriving jobs are dispatched. We focus on policies which base their decision only on the current states of the servers. For policy $\phi$, $N_k^\phi(t)$ denotes the state of server $k$ at time epoch $t$ and $\vec{N}^\phi(t) = (N_1^\phi(t),..., N_K^\phi(t) )$. 
Let $S_k^\phi(\vec{N}^\phi(t)) \in \{0,1\}$ represent whether or not an arrival is routed to server $k$ at time $t$ under policy $\phi$. Since a job can be dispatched to  at most one out of $K$ servers at each stage, we have  $\sum\limits_{k=1}^KS_k^\phi(\vec{N}) \le 1,$ which can be re-written as 
\begin{equation}\label{hard_constraint}
\sum\limits_{k=1}^K(1-S_k^\phi(\vec{N})) \ge K-1.
\end{equation}
If $\sum\limits_{k=1}^K  S_k^\phi(\vec{N})=0$, the job is blocked. 
The dynamics of the queue length process is then described as follows:
\begin{equation}\label{bandit_evolution}
N_k^\phi(t+1) = N_k^\phi(t)  +  S_k^\phi(\vec{N}^\phi(t))\psi(t+1)- R_k^{d_k}(t+1),
\end{equation}
	where $\psi(t)$ represents the number of jobs arriving in time period $t$ and $R_k^{d_k}(t)$ represents the process describing the number of jobs that departed in time period $t$  for server $k$ under the scheduling discipline LPS-${d_k}$. 
	
Let us denote by $\mathcal{U}$ the set of policies that satisfy the constraint (\ref{hard_constraint}) at each decision epoch and make the system ergodic. Throughout this paper, we assume that $p<\sum_{k=1}^K q_k$, which is the maximum stability condition. Hence,  under this assumption there exist feasible load balancing policies that make the system stable, i.e.,  $\mathcal{U} \neq \phi$.

\begin{rmk}\label{batch_arrival_extention}
In the model description, we assumed that at most one job can arrive per time period.  A natural generalization is to consider  uniformly bounded batch arrivals that need to be allocated to one server. We note here that the relaxation and decomposition technique that are developed in this paper will also go through in the batch setting, that is left out for the sake of tractability.
%works for both cases where either entire batch is allocated to one server or the batch is broken up and arrivals are dispatched one by one to the servers. In the former case, relaxation will consist of one dimensional problem with same batch arrivals whereas for the later, current index policy could be re-utilized (with at most one arrival in the server). 
\end{rmk}

\subsection*{Stochastic optimal control}
For server $k$, let $C_k(n)$ denote the cost of being in state $n$. We assume that $C_k(n)$ is non-decreasing and bounded by a polynomial of finite degree.   
Let $D$ be the cost of blocking a job. 
The objective is to find a scheduling policy, $\phi \in\mathcal{U}$, that minimizes the long run average-cost criterion,
 \begin{equation}\label{objective}
\mathcal{C^\phi} :=\limsup_{T\to\infty}\frac{1}{T} \mathbb{E}\sum\limits_{t=0}^T\left(\sum\limits_{k=1}^K C_k(N_k^\phi(t)) + pD \left(1- \sum\limits_{k=1}^KS_k^\phi(\vec{N}^\phi(t))\right)\right).
\end{equation}
 Due to the hard constraint (\ref{hard_constraint}), $\sum\limits_{k=1}^KS_k^\phi(\vec{N}^\phi(t))$ is either zero (when blocked) or one (when routed). Since $\limsup\limits_{T\uparrow\infty}~\frac{1}{T} \mathbb{E}\sum\limits_{t=0}^T$ $\left(1- \sum\limits_{k=1}^KS_k^\phi(\vec{N}^\phi(t))\right)$ represents the fraction of time that an arriving job is blocked, the second term in (\ref{objective}) represents the average cost of blocking. For very large blocking cost, i.e., $D\rightarrow\infty$, in order to minimize~(\ref{objective}), one will take  $\sum\limits_{k=1}^KS_k^\phi(\vec{N}^\phi(t))$ equal to 1. Thus, when setting $D=\infty$, one retrieves the load balancing model without blocking.  
 
 The above problem can be seen as a particular case of Markov decision process (MDP)  and an exact dynamic programming formulation is possible (see \cite{puterman2014markov}). For certain MDPs, it is feasible to explicitly characterize an optimal stochastic control. An important class of such problems is the multi-armed bandit problem (MABP). In a MABP, only one bandit can be made active and only the active bandit can change state while the state of all other bandits remain \textit{frozen}. In MABP, an optimal solution has a simple structure, known as Gittin's index policy (see \cite{gittins2011multi}). In brief, there exists functions $G_k (n_k)$, depending only on the parameters of bandit $k$, such that an optimal policy in state $\vec{n}$ is to serve the bandit having currently the highest index $G_k(n_k)$. 

However, the optimal scheduling policies for restless bandit problems is typically out of reach. We note that the transition probability matrix of $N_k^\phi(t)$ is action dependent. In particular the state of the server can evolve both when the job is dispatched to the server or not. Hence, each of the servers can be considered as a restless bandit, and the load balancing problem can be seen as a restless bandit problem. The analytical solution of restless problems is inaccessible because of the sample path constraint (\ref{hard_constraint}). Even the
numerical resolution of this problem via dynamic programming becomes quickly intractable because of the curse of
dimensionality.  The restless bandit problem has been reported to be PSPACE complete (see \cite{papadimitriou1999complexity}). 

Therefore, we deploy the relaxation approach pioneered by Whittle in \cite{whittle1988restless}. The main idea is to relax the hard constraint (\ref{hard_constraint}) where the constraint has to be respected only in average but not at every decision epoch. Whittle showed that under the so-called \emph{indexability property}, the solution to the relaxed problem is fully characterized by Whittle's index policy, that allocates an incoming job to the schedulers with index larger than the Lagrange multiplier associated to the relaxation. Whittle then proposed a heuristic for the original problem with hard constraint, in which an incoming job is dispatched to the server with the highest index. The latter heuristic is nowadays referred to as Whittle's index policy. Whittle's index policy is in general not optimal for the problem with hard constraint, however is known to be asymptotically optimal as the number of bandits grows to infinity (see \cite{weber1990index}), and it has been reported that its performance is nearly-optimal for different problems.

In the next section, we will carry out the above research agenda, by first considering the relaxed problem, and then establishing indexability. We will then obtain in Equation (\ref{general_index_value}) the main result of the paper, i.e., a closed-form expression for the Whittle Index.

%Therefore, we analyze the Whittle relaxation
%of the problem via the Lagrangian approach (see \cite{whittle1988restless}).  The Whittle relaxation develops a well performing heuristic for a stochastic control problem with sample path constraint (\ref{hard_constraint}) and results in the so-called index rules. These are greedy rules with dynamic nature, prescribing the following: Allocate the job to the scheduler with the highest index values. Index values are assigned to all the states of a server by the index function, which is furthermore independent of the other servers. Thus, index rules aim at decreasing the dimensionality of the problem by computing the index values for each server in isolation. Index rules are in general not optimal however Whittle's index policy is known to be asymptotically optimal (See \cite{weber1990index}). We now present relaxation and decomposition of the load-balancing problem.

%Before describing the relaxation and decomposition in Section \ref{relaxation}, we first present the load balancing problem which is the central theme of this work. We \textit{remark} that the index developed in Section \ref{relaxation} is not only valid for the load balancing problem but also for the bandits whose evolution is driven by a DTMC. 
\section{Derivation of Whittle's index}
\label{sec:derivationWhittle}
In this section we derive Whittle's index. In Section~\ref{relaxation} we describe how the relaxation  decomposes the original $K$ dimensional optimization problem into  $K$ one-dimensional subproblems. In Section~\ref{sec:threshold} we show that the optimal solution to such a subproblem is of threshold type, which allows us to show in Section~\ref{sec:indexability} that the relaxed problem satisfies the  \emph{indexability} property. The latter justifies the derivation of Whittle's index, which is stated in Section~\ref{indices}. An optimal policy for the relaxed problem,  which is described by Whittle's index, will then serve as a heuristic for the original optimization problem, as described in Section~\ref{sec:index}.

\subsection{Relaxation and decomposition}
\label{relaxation}
In~\cite{whittle1988restless}, Whittle proposed to replace the infinite set of sample-path constraints~\eqref{hard_constraint} by its time-average version, that is, on average at most $K-1$ servers are kept passive:
\begin{equation}\label{relaxed_constraint}
\limsup_{T\to\infty}\frac{1}{T} \mathbb{E}\left(\sum_{t=0}^T\sum\limits_{k=1}^K(1-S_k^\phi(\vec{N}^\phi(t)))\right)\le K-1.
\end{equation}
 We denote by $\mathcal{U}^{REL}$ the set of stationary policies for which the Markov chain is ergodic and that satisfy (\ref{relaxed_constraint}). We note that $\mathcal{U} \subset \mathcal{U}^{REL}$.  The objective of the relaxed problem is hence to minimize (\ref{objective}) among all policies in $\mathcal{U}^{REL}$. 
 
 The relaxed problem can be solved by considering the following unconstrained problem: find a policy $\phi$ that minimizes
\begin{equation}\label{relaxed_problem}
\limsup_{T\uparrow\infty}\frac{1}{T}\mathbb{E}\left[\sum_{t=0}^T \left(\sum_{k=1}^K\right.\left[C_k(N_k^\phi(t))  + p D\left(\frac{1}{K} - {S_k^\phi(\vec{N}^\phi(t))}\right)\right]\right. \left.\left.+ W\left( K-1-\sum_{k=1}^K(1-S_k^\phi(\vec{N}^\phi(t))) \right)\right)\right],
\end{equation}
where $W$ is the Lagrange multiplier. The key observation made by Whittle is that the above relaxed problem can be decomposed into $K$ subproblems, one for each server $k$, that is, minimize  
\begin{equation}\label{sub_problem}
\mathcal{C}_k^{\phi_k}:=\limsup_{T\uparrow\infty}\frac{1}{T}\mathbb{E}\left[\sum_{t=0}^T \left(C_k(N_k^{\phi_k}(t)) -(W-pD)( (1-S_k^{\phi_k}({N}^{\phi_k}(t)))\right)\right].
\end{equation}

The solution to (\ref{relaxed_problem}) is then obtained by combining the solution to the  $K$ separate subproblems (\ref{sub_problem}). Under a stationarity assumption, we can invoke ergodicity to show that (\ref{sub_problem}) is equivalent to minimizing 
\begin{equation}\label{ergodic_sub_problem}
\mathbb{E}(C_k(N_k^{\phi_k}, S_k^{\phi_k}({N}_k^{\phi_k}))-\left(W-pD\right) \mathbb{E}(\mathbf{1}_{S_k^{\phi_k}({N}_k^{\phi_k})=0}),
\end{equation}
where $N_k^{\phi_k}$ is distributed as the stationary distribution of the state of server $k$ under policy $\phi_k$. %Note that multiplier $W$ can be interpreted as subsidy for passivity. 

\subsection{Threshold optimality }
\label{sec:threshold}
 %In particular, a standard Lagrangian argument shows that there exists a value $W = W^*$  for which the constraint (\ref{relaxed_constraint}) is binding, i.e., an optimal policy that solves Problem (\ref{relaxed_problem}) for $W = W^*$ will on average activate (at most) $M$ bandits. 

In this section, we establish that the optimal solution of problem (\ref{sub_problem}) is of threshold type. That is, there is a threshold $n_k(W)$ such that when there are $n_k$ users in server $k$, $n_k \le n_k(W)$, then accepting a job is optimal, and otherwise blocking  a job is optimal. We let policy $\phi_k = n_k$ denote a threshold policy with threshold $n_k$.%, and we refer to it as 1-0 type. Following result establishes the threshold optimality of load balancing problem:
\begin{prop}\label{threshold_optimality}
Threshold policies are optimal for the relaxed load balancing problem, i.e., there exists an $n_k \in\{-1, 0,1, ...\}$ such that the threshold policy  $n_k$ optimally solves problem (\ref{ergodic_sub_problem}), for all $k=1,\ldots,K$.
\end{prop}
\begin{proof}
{
Let us drop the dependency on $k$ throughout the proof. Since there exists  $\phi \in \mathcal{U}_{REL}$, there exists a stationary optimal policy $\phi^*$ that optimally solves problem (\ref{ergodic_sub_problem}). Define $n^* = \inf\{ m\in\{ 0,1,...\}: S^{\phi^*}(m)=0 \}$. This implies $S^{\phi^*}(m)=~1,~\forall~m<n^*$ and $S^{\phi^*}(n^*)=0$. It follows from the evolution of the Markov chain (see Equation (\ref{bandit_evolution})) that all states $m > n^*$ are transient. Thus, $\pi^{\phi^*}(m) = 0~\forall~m>n^*$. Thus, the average cost as given by (\ref{ergodic_sub_problem}) under the optimal policy $\phi^*$ then reduces to 
\begin{eqnarray}\nonumber
\mathbb{E}(C(N^{\phi^*})-\left(W-pD\right) \mathbb{E}(\mathbf{1}_{S^{\phi^*}({N}^{\phi^*})=0})&=& \sum\limits_{m=0}^{n^*-1}C(m)\pi^{\phi^*}(m) + C(n^*)\pi^{\phi^*}(n^*) - \left(W-pD\right)\pi^{\phi^*}(n^*) \\\nonumber
&=&\mathbb{E}(C(N^{n^*}) - \left(W-pD\right)\pi^{n^*}(n^*),
\end{eqnarray}
that is, a threshold policy with threshold $n^*$ gives the optimal performance. 
 }
\end{proof}
The above proposition implies that an optimal policy for problem (\ref{ergodic_sub_problem}) is fully characterized by a threshold $n$. 
 We let $\pi_k^n(m)$ denote the steady state probability of being in state $m$ for bandit $k$ under the threshold policy $n$. Equation~\eqref{ergodic_sub_problem} under policy $\phi_k=n$ can be written as 
$$g_k^n(W):= \sum\limits_{m=0}^\infty C_k(m) \pi_k^n(m) - \left(W-pD\right)\sum\limits_{m=n+1}^\infty  \pi_k^n(m). $$
We hence conclude that the optimal solution of problem (\ref{ergodic_sub_problem}) is given by: 
$$g_k(W) = \min_{n}g_k^n(W).$$

\subsection{Indexability}
\label{sec:indexability}
 Indexability is the property that allows us to develop a heuristic for the original problem. This property requires to establish that as the Lagrange multiplier, or equivalently the subsidy for passivity, $W$, increases, the collection of states in which the optimal action is passive increases. It was first introduced by Whittle \cite{whittle1988restless} and we formalize it in the following definition.
\begin{defn}\label{indexability_defn}
A server is indexable if the set of states in which passive is an optimal action in (\ref{sub_problem}) (denoted by $D_k (W )$) increases in $W$, that is, $W' <W \Rightarrow D_k (W') \subseteq D_k(W )$.
\end{defn}
If indexability is satisfied, Whittle's index in state $N_k$ is defined as follows:
\begin{defn}
When server~$k$ is indexable, Whittle's index in state $N_k$ is defined as the smallest value for the subsidy such that actions active and passive are equally attractive in state $N_k$. The Whittle's index is denoted by $W_k(N_k)$.
\end{defn}

Given that the indexability property holds, Whittle established in \cite{whittle1988restless} that the solution to the relaxed control problem (\ref{relaxed_problem}) will be to activate all servers that are in state $n_k$ such that their
Whittle's index exceeds the Lagrange multiplier, i.e., $W_k (n_k) > W$. 

In order to prove indexability for our problem, we will make use of the following result 
 (see \ref{appendix:proofFCFS}  for its proof). 
 
 \begin{lem}\label{Lemma:stoch_dom}
$\sum\limits_{m=0}^n\pi_k^n(m)$ is non-decreasing in $n$. 
\end{lem}

We can now prove indexability of the problem.  Some ideas in this proof are adopted from \cite{glazebrook2009index}.

\begin{prop} The load balancing problem is Whittle indexable. 
\end{prop}
\begin{proof}
{
Let us drop the dependency on $k$ throughout the proof for ease of notation. Since an optimal policy for (\ref{ergodic_sub_problem}) is of threshold type, for a given subsidy $W$ the optimal average cost under threshold $n$ will be $g(W) := \min\limits_{n}\{g^{(n)}(W)\} $, where 
\begin{eqnarray}\nonumber
g^n(W) &:=& \sum\limits_{m=0}^\infty C(m) \pi^n(m) -\left(W-pD\right)\sum\limits_{m=n+1}^\infty  \pi^n(m)\\\nonumber
&=&\sum\limits_{m=0}^\infty C(m)\pi^n(m)+\left(W-pD\right)\sum\limits_{m=0}^n  \pi^n(m)-\left(W-pD\right).
\end{eqnarray}
Let $n(W)$ be the minimizer of $g(W)$. Note that $g^{(n)}(W)$ is an affine non-increasing function of $W$.  Thus, the function $g(W)$ is a lower envelope  of affine non-increasing functions of $W$. It thus follows that $g(W)$ is a concave non-increasing function. 

It directly follows that the right derivative of $g(W)$ in $W$ is given by $\sum\limits_{m=0}^{n(W)}\pi^{n(W)}(m)-1$. Since $g(W)$ is concave in $W$, the right derivative is non-increasing in $W$. But $\sum\limits_{m=0}^n\pi^n(m)$ is non-decreasing in $n$, from Lemma \ref{Lemma:stoch_dom}. It hence follows that $n(W)$ is non-increasing in $W$. Together with  $D(W) = \{m:m\ge n(W) \}$ and by  Definition  \ref{indexability_defn}, indexability follows. 
}
\end{proof}

\subsection{Whittle's index}\label{indices}
Recall that Whittle's index is the smallest value of $W$ such that we are indifferent of the action taken in state $n$. Under the optimality of threshold policies, one is indifferent of the action taken in state $n$ if the performance under threshold policies $n -1$ and $n$ are equal.
The following proposition presents a closed form expression of Whittle's index for the load-balancing problem.

\begin{prop}\label{proposition:monotone_index}
 Assume $\sum\limits_{m=0}^n\pi_k^n(m)$ is strictly increasing in $n$.  The Whittle index is given by 
\begin{equation}\label{general_index_value}
W_k(n) = pD-\frac{ \sum\limits_{m=0}^{n}C_k(m)[\pi_k^{n}(m)-\pi_k^{n-1}(m)] + C_k(n+1)\pi_k^{n}(n+1) }{\sum\limits_{m=0}^n\pi_k^n(m) - \sum\limits_{m=0}^{n-1}\pi_k^{n-1}(m)},
\end{equation}
if $W_k(n)$ is non-increasing in $n$. 
\end{prop}

Note that we require  $\sum\limits_{m=0}^n\pi_k^n(m)$ to be strictly increasing in $n$. In Lemma~\ref{Lemma:stoch_dom} we proved that this function is non-decreasing. Numerical evidence shows that the function is in fact strictly increasing. 

\begin{proof} We drop the subscript $k$ throughout the proof. 
Let $\tilde{W}(n)$ be the value for the subsidy such that the average cost under threshold policy $n$ is
equal to that under policy $n-1$. Hence, using (\ref{ergodic_sub_problem}), we obtain $\mathbb{E}(C(N^n))-(\tilde{W}(n)-pD) \mathbb{E}(\mathbf{1}_{S^n({N}^n)=0}) = \mathbb{E}(C(N^{n-1}))-(\tilde{W}(n)-pD) \mathbb{E}(\mathbf{1}_{S^{n-1}({N}^{n-1})=0})$, for all $n\ge 1$. Together with  $\mathbb{E}(\mathbf{1}_{S^n({N}^n)=0})  = \sum\limits_{m=n+1}^\infty \pi^n(m) = 1-\sum\limits_{m=0}^n\pi^n(m)$, and the fact that $\sum\limits_{m=0}^n\pi_k^n(m)$ is strictly increasing in $n$, we obtain 
$$\tilde{W}(n) = {pD}-\frac{\mathbb{E}(C(N^n)) - \mathbb{E}(C(N^{n-1}))}{\sum\limits_{m=0}^n\pi^n(m) - \sum\limits_{m=0}^{n-1}\pi^{n-1}(m)}.$$
Below we will show that  $\tilde{W}(n)$ is non-increasing in $n$. This then  implies that $g(\tilde{W}(n)) = g^{(n)}(\tilde{W}(n)) = g^{(n-1)}(\tilde{W}(n))$. Now, recall that  $\frac{dg^{(n)}(W)}{dW} = - \sum\limits_{m=n+1}^\infty\pi^n(m)$ is non-decreasing in $n$. Hence, $g(W ) = g^{(n)}(W)$, for $\tilde{W}(n) \leq W \leq\tilde{W}(n - 1)$. This implies that Whittle's index is given by $W(n) = \tilde{W}(n)$.  Note that $\mathbb{E}(C(N^n))  = \sum\limits_{m=0}^{n+1}C(m)\pi^{n}(m)$, hence $\tilde{W}(n)$  simplifies to the Whittle's index as stated in the proposition.
\end{proof}

\begin{rmk}
For a load balancing problem with homogeneous PS servers,  an iterative scheme to approximate Whittle's  index was  reported in~\cite{borkar2017whittle}. The iterative scheme to compute Whittle's index in \cite{borkar2017whittle} adjusts the current guess for the index in the direction of decreasing the discrepancy in the active and passive values which should agree for the exact index. Further, a linear interpolation is used after computing the index for sufficiently large number of states which makes the indices approximate in nature in~\cite{borkar2017whittle}. 
%\item The closed form expressions for Whittle's index have been reported previously in load balancing problem for FCFS schedulers with abandonment in \cite{glazebrook2009index} and with dedicated arrivals in \cite{argon2009dynamic}.
\end{rmk}

Our approach results in a closed form expression for Whittle's index in terms of steady-state probabilities for LPS-$d$ schedulers. Though, we could not theoretically argue the non-increasing nature of the index~(\ref{general_index_value}),  we numerically note that this is indeed the case for a wide set of parameters and for different cost functions (see Figures \ref{pattern_indices}, \ref{Indices_pattern_sameq} and \ref{Indices_pattern_diffq}). We note that Whittle's index~\eqref{general_index_value} has a common term $pD$. In case no blocking is allowed, Whittle's index is obtained by simply dropping the term $pD$.

%Based on Whittle's index, we now provide the details of a heuristic for load-balancing problem with sample path constraint (\ref{hard_constraint}). 

\subsection{Heuristic for load-balancing problem} \label{sec:index}
In this section, we describe how the optimal solution to the relaxed optimization problem is used to obtain a heuristic for the original model. The optimal solution to the relaxed problem, that is, activate all servers that are in a state $n_k$ such that $W_k(n_k) > W$, might be infeasible for the original model where a job can be dispatched to at most one server. Hence, Whittle \cite{whittle1988restless} proposed the following heuristic, which is referred as Whittle's index policy.

\begin{defn}[\textbf{Whittle's index policy}] Assume at time $t$ we are in state $\vec{N}(t) = \vec{n}$. 
\begin{itemize}
\item
In case $D<\infty$ and hence blocking is allowed, Whittle's index policy sends a new arriving job to the server having currently the \textit{highest} non-negative Whittle's index value $W_k (n_k )$. If all Whittle's indices are negative, the job is blocked. 
\item In case no blocking is allowed (and hence $D=\infty$), Whittle's index policy sends a new arriving job to the server having currently the \textit{highest} (possibly negative) Whittle's index value $W_k (n_k )$. \end{itemize}
\end{defn}

In case blocking is allowed, and  all servers are in a state such that their Whittle's index is negative, all servers are kept passive, i.e., the job is not dispatched to any of the servers and it is blocked. The latter is a direct consequence of the relaxed optimization problem: when the Whittle's index is negative for a server in state $\bar{n}$, this means that it is made active only if $W < W_k (\bar{n}) < 0$, that is, when a cost is paid for being passive.

\begin{rmk}
The closed form expressions of Whittle's index in Proposition \ref{proposition:monotone_index} facilitate a major computational saving on the computation of the true optimal policy.
\end{rmk}

\section{Computation of Whittle's index}\label{indices_for_different_discipline}

Whittle's index depends on the stationary distribution of the threshold policies, see Proposition \ref{proposition:monotone_index}. 
In this section we give further details on how these stationary probabilities can be calculated. In general, the closed form expression of the stationary distribution is  intractable. 
In Section~\ref{LPS-d}, we present  details of the transition probability matrix and how this can be used to  compute Whittle's index numerically. In Section~\ref{FCFS_scheme}, we focus on $d=1$ (FCFS), in which case a closed form expression for the steady state distribution, and hence for Whittle's index, can be given.

\subsection{LPS-$d$ discipline}\label{LPS-d}
In this section we describe the transitions of the one-dimensional problem of the relaxed optimization problem. Hence, we are given one single server, which implements LPS-$d_k$, and who accepts customers until threshold $n$. The one-step state transitions of this DTMC of server $k$ are shown in Figure \ref{fig:LPS_scheme} for  a state  $m<n$.
Recall that $q_k^{d_k}(i|m)$ was defined in~\eqref{LPSD_prob} and describes the probability of $i$ departures when $m$ jobs are in server $k$.   

\begin{figure}[ht]
\centering
\resizebox{.47\textwidth}{!}{% Graphic for TeX using PGF
% Title: /home/manu/Dropbox/CIMI-Toulouse-Manu-Urtzi-Maaike/restless_bandits/Paper_submissions/PERFORMANCE/Mechanism_LPS.dia
% Creator: Dia v0.97.3
% CreationDate: Mon May 14 19:03:26 2018
% For: manu
% \usepackage{tikz}
% The following commands are not supported in PSTricks at present
% We define them conditionally, so when they are implemented,
% this pgf file will use them.
\ifx\du\undefined
  \newlength{\du}
\fi
\setlength{\du}{15\unitlength}
\begin{tikzpicture}
\pgftransformxscale{1.000000}
\pgftransformyscale{-1.000000}
\definecolor{dialinecolor}{rgb}{0.000000, 0.000000, 0.000000}
\pgfsetstrokecolor{dialinecolor}
\definecolor{dialinecolor}{rgb}{1.000000, 1.000000, 1.000000}
\pgfsetfillcolor{dialinecolor}
\definecolor{dialinecolor}{rgb}{1.000000, 1.000000, 1.000000}
\pgfsetfillcolor{dialinecolor}
\pgfpathellipse{\pgfpoint{23.071664\du}{10.204972\du}}{\pgfpoint{1.478364\du}{0\du}}{\pgfpoint{0\du}{1.301682\du}}
\pgfusepath{fill}
\pgfsetlinewidth{0.100000\du}
\pgfsetdash{}{0pt}
\pgfsetdash{}{0pt}
\pgfsetmiterjoin
\definecolor{dialinecolor}{rgb}{0.000000, 0.000000, 0.000000}
\pgfsetstrokecolor{dialinecolor}
\pgfpathellipse{\pgfpoint{23.071664\du}{10.204972\du}}{\pgfpoint{1.478364\du}{0\du}}{\pgfpoint{0\du}{1.301682\du}}
\pgfusepath{stroke}
% setfont left to latex
\definecolor{dialinecolor}{rgb}{0.000000, 0.000000, 0.000000}
\pgfsetstrokecolor{dialinecolor}
\node at (23.071664\du,10.399972\du){\large $m$};
\definecolor{dialinecolor}{rgb}{1.000000, 1.000000, 1.000000}
\pgfsetfillcolor{dialinecolor}
\pgfpathellipse{\pgfpoint{17.993364\du}{10.063332\du}}{\pgfpoint{1.478364\du}{0\du}}{\pgfpoint{0\du}{1.301682\du}}
\pgfusepath{fill}
\pgfsetlinewidth{0.100000\du}
\pgfsetdash{}{0pt}
\pgfsetdash{}{0pt}
\pgfsetmiterjoin
\definecolor{dialinecolor}{rgb}{0.000000, 0.000000, 0.000000}
\pgfsetstrokecolor{dialinecolor}
\pgfpathellipse{\pgfpoint{17.993364\du}{10.063332\du}}{\pgfpoint{1.478364\du}{0\du}}{\pgfpoint{0\du}{1.301682\du}}
\pgfusepath{stroke}
% setfont left to latex
\definecolor{dialinecolor}{rgb}{0.000000, 0.000000, 0.000000}
\pgfsetstrokecolor{dialinecolor}
\node at (17.993364\du,10.258332\du){\large $m-1$};
\pgfsetlinewidth{0.100000\du}
\pgfsetdash{}{0pt}
\pgfsetdash{}{0pt}
\pgfsetmiterjoin
\pgfsetbuttcap
{
\definecolor{dialinecolor}{rgb}{0.000000, 0.000000, 0.000000}
\pgfsetfillcolor{dialinecolor}
% was here!!!
\pgfsetarrowsend{to}
\definecolor{dialinecolor}{rgb}{0.000000, 0.000000, 0.000000}
\pgfsetstrokecolor{dialinecolor}
\pgfpathmoveto{\pgfpoint{22.505900\du}{11.407600\du}}
\pgfpathcurveto{\pgfpoint{21.750000\du}{12.756700\du}}{\pgfpoint{19.400000\du}{13.056700\du}}{\pgfpoint{17.993400\du}{11.365000\du}}
\pgfusepath{stroke}
}
\pgfsetlinewidth{0.100000\du}
\pgfsetdash{}{0pt}
\pgfsetdash{}{0pt}
\pgfsetmiterjoin
\pgfsetbuttcap
{
\definecolor{dialinecolor}{rgb}{0.000000, 0.000000, 0.000000}
\pgfsetfillcolor{dialinecolor}
% was here!!!
\pgfsetarrowsend{to}
\definecolor{dialinecolor}{rgb}{0.000000, 0.000000, 0.000000}
\pgfsetstrokecolor{dialinecolor}
\pgfpathmoveto{\pgfpoint{22.320596\du}{9.030600\du}}
\pgfpathcurveto{\pgfpoint{20.663496\du}{6.439270\du}}{\pgfpoint{23.600000\du}{5.356650\du}}{\pgfpoint{23.250000\du}{8.706650\du}}

\pgfpathmoveto{\pgfpoint{20.320596\du}{12.40600\du}}
\pgfpathcurveto{\pgfpoint{19.663496\du}{9.439270\du}}{\pgfpoint{23.600000\du}{17.356650\du}}{\pgfpoint{23.250000\du}{13.706650\du}}

\pgfusepath{stroke}
}
\pgfsetlinewidth{0.100000\du}
\pgfsetdash{}{0pt}
\pgfsetdash{}{0pt}
\pgfsetmiterjoin
\pgfsetbuttcap
{
\definecolor{dialinecolor}{rgb}{0.000000, 0.000000, 0.000000}
\pgfsetfillcolor{dialinecolor}
% was here!!!
\pgfsetarrowsend{to}
\definecolor{dialinecolor}{rgb}{0.000000, 0.000000, 0.000000}
\pgfsetstrokecolor{dialinecolor}
\pgfpathmoveto{\pgfpoint{24.437500\du}{9.706840\du}}
\pgfpathcurveto{\pgfpoint{25.187500\du}{8.506840\du}}{\pgfpoint{25.850000\du}{8.406650\du}}{\pgfpoint{27.650000\du}{9.006650\du}}
\pgfusepath{stroke}
}
\definecolor{dialinecolor}{rgb}{1.000000, 1.000000, 1.000000}
\pgfsetfillcolor{dialinecolor}
\pgfpathellipse{\pgfpoint{10.393364\du}{10.161682\du}}{\pgfpoint{1.478364\du}{0\du}}{\pgfpoint{0\du}{1.301682\du}}
\pgfusepath{fill}
\pgfsetlinewidth{0.100000\du}
\pgfsetdash{}{0pt}
\pgfsetdash{}{0pt}
\pgfsetmiterjoin
\definecolor{dialinecolor}{rgb}{0.000000, 0.000000, 0.000000}
\pgfsetstrokecolor{dialinecolor}
\pgfpathellipse{\pgfpoint{10.393364\du}{10.161682\du}}{\pgfpoint{1.478364\du}{0\du}}{\pgfpoint{0\du}{1.301682\du}}
\pgfusepath{stroke}
% setfont left to latex
\definecolor{dialinecolor}{rgb}{0.000000, 0.000000, 0.000000}
\pgfsetstrokecolor{dialinecolor}
\node at (10.393364\du,10.356682\du){\large $m-{d_k}$};
\pgfsetlinewidth{0.100000\du}
\pgfsetdash{}{0pt}
\pgfsetdash{}{0pt}
\pgfsetmiterjoin
\pgfsetbuttcap
{
\definecolor{dialinecolor}{rgb}{0.000000, 0.000000, 0.000000}
\pgfsetfillcolor{dialinecolor}
% was here!!!
\pgfsetarrowsend{to}
\definecolor{dialinecolor}{rgb}{0.000000, 0.000000, 0.000000}
\pgfsetstrokecolor{dialinecolor}
\pgfpathmoveto{\pgfpoint{22.424611\du}{11.428875\du}}
\pgfpathcurveto{\pgfpoint{21.103011\du}{13.928875\du}}{\pgfpoint{12.634565\du}{15.471498\du}}{\pgfpoint{10.927965\du}{11.428198\du}}

\pgfusepath{stroke}
}
\pgfsetlinewidth{0.100000\du}
\pgfsetdash{{\pgflinewidth}{0.200000\du}}{0cm}
\pgfsetdash{{\pgflinewidth}{0.200000\du}}{0cm}
\pgfsetbuttcap
{
\definecolor{dialinecolor}{rgb}{0.000000, 0.000000, 0.000000}
\pgfsetfillcolor{dialinecolor}
% was here!!!
\definecolor{dialinecolor}{rgb}{0.000000, 0.000000, 0.000000}
\pgfsetstrokecolor{dialinecolor}
\draw (13.400000\du,10.005000\du)--(15.700000\du,10.005000\du);
}

\definecolor{dialinecolor}{rgb}{1.000000, 1.000000, 1.000000}
\pgfsetfillcolor{dialinecolor}
\pgfpathellipse{\pgfpoint{6.593364\du}{10.206682\du}}{\pgfpoint{1.522250\du}{0\du}}{\pgfpoint{0\du}{1.340323\du}}
\pgfusepath{fill}
\pgfsetlinewidth{0.100000\du}
\pgfsetdash{}{0pt}
\pgfsetdash{}{0pt}
\pgfsetmiterjoin
\definecolor{dialinecolor}{rgb}{0.000000, 0.000000, 0.000000}
\pgfsetstrokecolor{dialinecolor}
\pgfpathellipse{\pgfpoint{6.593364\du}{10.206682\du}}{\pgfpoint{1.522250\du}{0\du}}{\pgfpoint{0\du}{1.340323\du}}
\pgfusepath{stroke}
% setfont left to latex
\definecolor{dialinecolor}{rgb}{0.000000, 0.000000, 0.000000}
\pgfsetstrokecolor{dialinecolor}
\node at (6.593364\du,10.401682\du){\scriptsize $m-{d_k}-1$};
\definecolor{dialinecolor}{rgb}{1.000000, 1.000000, 1.000000}
\pgfsetfillcolor{dialinecolor}
\pgfpathellipse{\pgfpoint{28.008364\du}{10.311682\du}}{\pgfpoint{1.478364\du}{0\du}}{\pgfpoint{0\du}{1.301682\du}}
\pgfusepath{fill}
\pgfsetlinewidth{0.100000\du}
\pgfsetdash{}{0pt}
\pgfsetdash{}{0pt}
\pgfsetmiterjoin
\definecolor{dialinecolor}{rgb}{0.000000, 0.000000, 0.000000}
\pgfsetstrokecolor{dialinecolor}
\pgfpathellipse{\pgfpoint{28.008364\du}{10.311682\du}}{\pgfpoint{1.478364\du}{0\du}}{\pgfpoint{0\du}{1.301682\du}}
\pgfusepath{stroke}
% setfont left to latex
\definecolor{dialinecolor}{rgb}{0.000000, 0.000000, 0.000000}
\pgfsetstrokecolor{dialinecolor}
\node at (28.008364\du,10.506682\du){\large $m+1$};
% setfont left to latex
\definecolor{dialinecolor}{rgb}{0.000000, 0.000000, 0.000000}
\pgfsetstrokecolor{dialinecolor}
\node[anchor=west] at (15.500000\du,6.200000\du){\large $q_k^{d_k}(0|m)(1-p)+q_k^{d_k}(1|m)p$};
% setfont left to latex
\definecolor{dialinecolor}{rgb}{0.000000, 0.000000, 0.000000}
\pgfsetstrokecolor{dialinecolor}
\node[anchor=west] at (22.00000\du,12.850000\du){\large ${ q_k^{d_k}(1|m)(1-p)+q_k^{d_k}(2|m)p}$};
% setfont left to latex
\definecolor{dialinecolor}{rgb}{0.000000, 0.000000, 0.000000}
\pgfsetstrokecolor{dialinecolor}
\node[anchor=west] at (13.0000\du,14.500000\du){\large $q_k^{d_k}({d_k}|m)(1-p)$};
% setfont left to latex
\definecolor{dialinecolor}{rgb}{0.000000, 0.000000, 0.000000}
\pgfsetstrokecolor{dialinecolor}
\node[anchor=west] at (24.6000\du,8.0000\du){\large $q_k^{d_k}(0|m)p$};
\end{tikzpicture}}
\caption{One step state evolution of Markov chain in LPS-$d_k$ scheduling scheme for state $m\le n$.}\label{fig:LPS_scheme}
\end{figure}
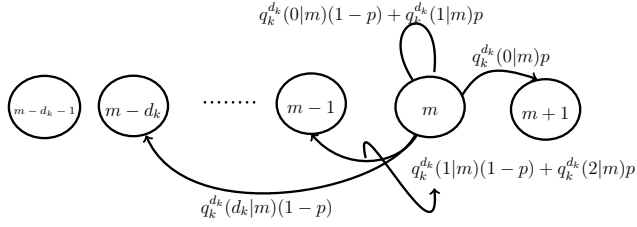 
The transition probability from state $i$ to state $j$, $p_{i,j}$, for LPS-$d_k$ scheme is given by:
\[
 p_{i, j} = 
  \begin{cases} 
   pq_k^{d_k}(0|i) & \text{if } j = i+1, \\
   pq_k^{d_k}(i-j+1|i) + (1-p)q_k^{d_k}(i-j|i)    & \text{if } j = i,i-1,\dots, i-\min \{i,d_k\}+1, \\
   (1-p)q_k^{d_k}( \min \{i,d_k\}|i)    & \text{if } j = i-\min \{i,d_k\},\\
   0       & \text{ Otherwise}.
  \end{cases}
\]
Obtaining a  closed-form expression for the  stationary distribution for arbitrary~$d_k$ seems infeasible, and hence, we will compute the stationary probabilities and Whittle's index numerically in Section~\ref{numerics}. In the next section we will consider the case $d_k=1$ and derive a closed-form expression.

\subsection{FCFS scheduling discipline}\label{FCFS_scheme}

%The stationary distribution is tractable for FCFS scheduling scheme, which allows us  give the closed form expression of Whittle's index for the FCFS scheduler.

Under FCFS, a job in front of the queue departs with probability $q_k$ in a given time slot. Thus, the evolution of DTMC for server $k$ is according to the following transition probability matrix under the threshold policy $n$:
$$\mathbb{P}_{k} = 
 \begin{pmatrix}
1-p & p &0 & 0 & \cdots & 0 & 0 &0 \\
  d_k & r_k &b_k & 0 & \cdots & 0 & 0 &0 \\
0 &  d_k & r_k &b_k &  \cdots & 0 & 0 &0 \\
  \vdots  & \vdots  & \vdots  & \vdots &\ddots &\vdots&\vdots&\vdots \\
0 & 0 & 0 & 0 &  \cdots & r_k & b_k &0 \\ 
0 & 0 & 0 & 0 &  \cdots & d_k & r_k & b_k \\
0 & 0 & 0 & 0 &  \cdots & 0 & q_k &1-q_k \\
 \end{pmatrix}_{(n+2)\times (n+2)},
$$
where $b_k = p(1-q_k),~d_k = q_k(1-p)$ and $r_k = pq_k + (1-p)(1-q_k)$. 

The stationary distribution  $\pi_k(\cdot)$ of the above DTMC is given by solving a set of linear equations $\mathbb{\pi}_k^n = \mathbf{\pi}^n_k\mathbb{P}_k$ with the normalizing condition ($\sum\limits_{i} \pi_k^n(i) = 1$). The solution can be written in closed form and is given by:
%for the stationary distribution of being in state $m$ under the threshold policy $n$ for bandit $k$ be denoted by $\pi_k^n(m)$. The stationary distribution for FCFS simplifies in a closed form expression:  
\begin{eqnarray}\label{stat_prob_0}
\pi_k^n(0) &=& \frac{d_k}{p}  \pi_k^n(1), \\\label{stat_prob_1}
\pi_k^n(1)&=& \frac{\frac{p}{q_k}\left(1-\frac{p}{q_k}\right)(1-p)^{n-1}}{(1-p)^n - \left(\frac{p}{q_k}\right)^{n+2}(1-q_k)^n},\\\label{stat_prob_m}
\pi_k^n(m)&=& \left(\frac{b_k}{d_k}\right)^{m-1}\pi_k^n(1),~m=2,3,...,n,\\\label{stat_prob_n}
\pi_k^n(n+1)&=& \frac{b_k}{q_k}\left(\frac{b_k}{d_k}\right)^{n-1}\pi_k^n(1).
\end{eqnarray}
%\textbf{Mean number of jobs in the system:} The mean number in system under the threshold policy $n$, $\mathbb{E}^n(N)$, will be given by:
%$$\mathbb{E}^n(N) = \sum_{i=0}^{n+1}i\pi^n(i),$$
%Above simplifies to:  
%\begin{equation}\nonumber
%\mathbb{E}^n(N) = \left[ \frac{1-(n+1)\left(\frac{b}{d}\right)^{n} + n\left(\frac{b}{d}\right)^{n+1}  }{\left(1-\frac{b}{d}\right)^2} + \left(\frac{b}{q}\right)\left(\frac{b}{d}\right)^{n-1}\right]\frac{\frac{p}{q}\left(1-\frac{p}{q}\right)(1-p)^{n-1}}{(1-p)^n - \left(\frac{p}{q}\right)^{n+2}(1-q)^n}. 
%\end{equation}
%Taking the threshold $n \rightarrow \infty$, the mean number simplifies to:
%\begin{equation}\nonumber
%\mathbb{E}(N) = \frac{\frac{p}{q}(1-p)}{1-\frac{p}{q}}.
%\end{equation} 
We can now check the condition needed in supposition of Proposition \ref{proposition:monotone_index}, i.e., $\sum\limits_{i=0}^n\pi_k^n(i)$ is strictly increasing in $n$, or equivalently,  $\pi_k^n(n+1)< \pi_k^{n-1}(n)$. After some algebra, the latter simplifies to 
$\frac{b_k}{d_k}\pi_k^n(1) < \pi_k^{n-1}(1).$
Using the stationary probabilities, it can be easily verified that the latter holds true. In addition, we will prove that the expression~\eqref{general_index_value} in Proposition \ref{proposition:monotone_index} is non-increasing. This provides us with a closed form expression for Whittle's index under FCFS  (see Appendix A.2. for the proof). 

\begin{prop}\label{prop:FCFS}
Whittle's index, $W_k(n)$, for bandit $k$ under FCFS is given by
\begin{eqnarray}\nonumber
&&pD+ \frac{p^2}{q_k^2(1-p)}\sum_{m=0}^{n-1}C_k(m)\left(\frac{b_k}{d_k}\right)^{m-1} -\frac{C_k(n)q_k}{(q_k-p)}\left[p + \left(\frac{p}{q_k}\right)^{n+1}\left(\frac{1-q_k}{1-p}\right)^{n-1}\left(\frac{p}{q_k}-p-1 \right) \right] 
\\
&&\hspace{3cm}
\label{Index_FCFS}
  -\frac{C_k(n+1)p(1-q_k)}{(q_k-p)}\left[1- \left(\frac{p}{q_k}\right)^{n+1}\left(\frac{1-q_k}{1-p}\right)^{n-1} \right],
\end{eqnarray}
which is a non-increasing function.
In addition, $$W_k(0)= pD + C_k(1)p\frac{q_k+p-1}{q_k},$$ and
if $\lim_{n\to\infty}C_k(n)\to \infty,$  then $\lim_{n\uparrow\infty}W_k(n)\rightarrow-\infty.$
\end{prop}

\begin{rmk}
Load balancing problems with \textit{FCFS schedulers} have been previously modeled as  multi-armed restless bandit problems and closed form expressions for Whittle's index are available in the literature for the  \textit{continuous} time setting. For example, Ni\~{n}o-mora studied  a load balancing problem with finite buffer queues and derived a closed form expression (see \cite[section 8.1]{nino2002dynamic}) whereas  \cite{argon2009dynamic} and \cite{glazebrook2009index} derived Whittle's index for load-balancing problems with dedicated arrivals and abandonments, respectively.  Note that these indices are different from ours, as we consider a discrete-time setting.
\end{rmk}

%\ \\
%{\tt 2.  MAKE REFERENCE WITH EXISTING WORK ON FCFS}
%
%Closed form expressions for Whittle's indices have been previously derived for 
%
%
%\begin{itemize}
%\item Nino-mora \cite{NM1} considers the load-balancing with finite buffer queues. 
%\item Nino-mora \cite{NM2} considers countable state space for the bandits but for make to order/make to stock M/G/1 queue. 
%\item Glazebrook et. al. consider the load-balancing with abandonment in FCFS queues. If we take abandonment parameter to be 0, index becomes constant. 
%\item Nilay et. al. \cite{CEF} consider the load-balancing with dedicated arrivals. If we consider the parameter for dedicated arrivals ($\eta_i = 0$), we get the load balancing with FCFS arrivals and the index in Equation (35) of \cite{} simplifies to 
%$$W_k(n) = \frac{(1-\beta_k^{n+1})}{1-\beta_k}C_k(n) - \sum_{j=0}^{n-1}\beta^jC_k(j)$$
%where $\beta_k = \frac{\lambda}{\mu_k}$. 
%\end{itemize}   
%
%
%In the case of linear holding costs $h_k (n ) =
%h kn$ , under the long-run average criterion, the routing
%index in  Nino-mora \cite{NM1}:
%
%$$W_k(n) = \frac{h_k}{\mu_k}\left[ \frac{\rho_k^{n+2} -1}{(\rho_k-1)^2 }- \frac{n+2}{\rho_k -1} \right]$$
%where $\rho_k = \lambda/\mu_k$. 
%
%Our index for FCFS and liner cost criterion is:
%
%\begin{equation}\label{linear_cost_FCFS_index}
%W_k(n) = pD+ \frac{Cp^2 (1-p)}{(q_k-p)^2} -\frac{Cp(1-q_k)}{q_k-p} - \frac{nCp}{q_k-p}- \frac{Cp^3(1-p)}{q_k(q_k-p)^2}\left(\frac{p(1-q_k)}{q_k(1-p)}\right)^{n}.
%\end{equation}	
%

Recall that Whittle's index policy is proposed as an efficient heuristic. It is however not known in general whether Whittle's index policy will indeed make the system stable. However, an immediate consequence  from $\lim_{n\uparrow\infty}W_k(n)\rightarrow-\infty$, is that the Whittle index policy is stable for the systems with blocking: As the Whittle index (for any $d$ and $q$) approaches to $-\infty$, there is a finite threshold for each server above which  Whittle's index policy will block new arriving jobs. 

Without blocking, every job needs to be routed to some server.  Numerically we observed that the optimal action in Whittle's index policy  is determined by some linear switching curve (see Figures \ref{switching_dssd}, \ref{switching_ssdd} and \ref{homogeneous_queues}). That is, when $N_k$ is below some function which is linear in $N_i$, $i\not= k$, then we send to server~$k$.   
One would therefore expect that Whittle's index policy is maximum stable as it will not dispatch jobs to relatively crowded servers. 
In~\cite{ZWTSS17} maximum stability was proved for a set of load balancing policies that made sure that smaller queues were sent more traffic. %{\color{red}One thus expects that Whittle's index policy is maximum stable as it will not dispatch jobs to relatively crowded servers.} %{\tt {\color{blue} Given that the switching curve is a linear function, we expect to prove, using similar ideas as in~\cite{ZWTSS17} that Whittle's index  policy is maximum stable.}}

\subsubsection{Linear cost criterion}\label{FCFS_linear_cost_index}
For the linear cost function, that is, $C_k(m) = mC_k$, the Whittle index (as stated in Proposition \ref{prop:FCFS}) can be rewritten as 
\begin{equation}\label{linear_cost_FCFS_index}
W_k(n) = pD+ \frac{C_kp^2 (1-p)}{(q_k-p)^2} -\frac{C_kp(1-q_k)}{q_k-p} - \frac{nC_kp}{q_k-p}- \frac{C_kp^3(1-p)}{q_k(q_k-p)^2}\left(\frac{p(1-q_k)}{q_k(1-p)}\right)^{n}.
\end{equation}
The details  can be found in \ref{appendix:proofFCFS}.  

In case $p<q_k$, it follows from~\eqref{linear_cost_FCFS_index} that   $W(n)$ behaves as 
$$\frac{W(n)}{n} = C_k   \frac{p}{q_k-p} +o(1), \mbox{ for large $n$}.$$
That is, when there are many jobs in the servers, under Whittle's index policy, the job is routed to the server having highest value $C_k n_k \frac{p}{q_k-p}$, with $n_k$ as number of jobs in the server~$k$. 
Instead, under JSEW, the job is routed to the server with the least expected workload, or in other words, to the server with the largest $ n_k/q_k$ index. 
Hence, JSEW is greedy, and only sees what lies ahead of an incoming job, i.e., it ignores the impact of jobs in the future. Whittle's index policy differs from JSEW (in case $C_k=1$) by the multiplicative term $p/ (1- p/q)$, that is, it also takes into account the impact of future arrivals $p$. In the numerical section we will see that Whittle's index policy indeed performs better than JSEW. %{\color{blue}{\tt ANY COMMENT ON INTERPRETATION FOR THE TERM $p/ (1- p/q)$? (Ask Urtzi)}}

%\subsection*{Properties of the index}{\color{blue}
%{\tt NOT SURE WHETHER WE SHOULD KEEP THE BELOW PROPERTIES (Ask Urtzi)}
%\begin{enumerate}
%\item In case of $q_1>q_2$, it can be easily argued from Equation (\ref{linear_cost_FCFS_index}) that $W_1(n) > W_2(n)$. Thus, Whittle's index policy will dispatch the job to the faster queue when the number of jobs are the same in both queues.  	
%
%\item  \textbf{Light traffic indices:} The light traffic corresponds to the case when $p\rightarrow 0 $. By using equation (\ref{linear_cost_FCFS_index}), the limiting index in the light traffic is given by:
%
%$$\lim_{p\rightarrow 0 }\frac{1}{p} W_k(n) \rightarrow D  -\frac{C(1-q_k)+nC}{q_k}$$

%Thus, the limiting behavior of light traffic indices as $n\rightarrow \infty$ is same as JSEW ($C/q_k$). {\tt NOT TRUE, BECAUSE YOU MULTIPLIED STILL BY $p$. }
%\item $p\rightarrow q_k $, let $f^1(p) = (q_k-p)^2$, 
%
%
%$$\lim_{p\rightarrow q_k }f^1(p)W_k(n) \rightarrow  0$$
%
%$$\lim_{p\rightarrow 1 }f^1(p)W_k\left(\frac{n}{q_k-p}\right) \rightarrow - nCp$$

%\end{enumerate}}

\section{Other load balancing policies}\label{scheduling_schemes}
In the numerical section we compare the performance of Whittle's index policy with various dispatching policies that we describe here.
%different scheduling schemes. We first provide the details for the computation of an optimal policy using value iteration, and then  discuss some other standard dispatching rules.  

\subsection{Optimal load balancing policy}\label{optimal_policy}
We numerically solve Bellman's optimality equation by value iteration to compute the optimal performance. Consider a system with $K$ servers and let $m_1, m_2, \cdots, m_K$ denote the number of jobs in these servers. Bellman's optimality equation for the average cost criterion is given by:
\begin{equation}\label{optimality_eqn}
g+ V_{t+1}(m_1, m_2, \cdots, m_K) =\sum\limits_{i=1}^KC_i(m_i) +  \min_{i \in \{1,2,\dots, K\} } \left\lbrace \mathbb{E}^i(V_t(.)|m_1, m_2, \cdots, m_K) \right\rbrace,
\end{equation}
where $\mathbb{E}^i(.)$ denotes the expectation with respect to the transition probability when an arriving job is dispatched to server $i$ and $g$ represents the average cost under an optimal policy. An optimal server to send a new arriving job to, will be a server that minimizes the RHS of \eqref{optimality_eqn}.
In the numerical section, we will focus on $K=2$. The term $\mathbb{E}^i(V_t(.)|m_1, m_2)$ then simplifies to 
$$\mathbb{E}^i(V_t(.)|m_1, m_2) = p \sum_{i_1=0}^{m_1} \sum_{i_2=0}^{m_2}  q_1^{d_1}(i_1|m_1)q_2^{d_2}(i_2|m_2)V_t(m_1+\mathbf{1}_{\{\mathcal{A}=1\}} - i_1, m_2+\mathbf{1}_{\{\mathcal{A}=2\}} - i_2)$$ $$ + (1-p) \sum_{i_1=0}^{m_1} \sum_{i_2=0}^{m_2} q_1(i_1)q_2(i_2)V_t(m_1 - i_1, m_2 - i_2), $$
where $q_l^d(i|m)$ is given in~\eqref{LPSD_prob} and denotes the probability of $i$ jobs departing in one time slot, when there are $m$  jobs present in server $l$ with LPS-$d$ service discipline and $\mathbf{1}_{\{\mathcal{A}=j\}},~j=1,2$ presents the indicator function of the event that an arrival is dispatched to queue $j$. Due to the curse of dimensionality, the optimality equation~\eqref{optimality_eqn} cannot be solved for a moderate number of servers. We therefore use the value iteration algorithm to find the optimal policy numerically, which we describe below:

\iffalse
Using above, the Bellman's optimality equation simplifies to:
$$g+ V_{t+1}(m_1, m_2) = \sum\limits_{i=1}^2C_i(m_i) + (1-p) \sum_{n_1=0}^{m_1} \sum_{n_2=0}^{m_2} q_1(n_1)q_2(n_2)V_t(m_1 - n_1, m_2 - n_2)
$$
\begin{equation}\label{optimality_eqn}
+  p\min_{i \in \{1,2\} } \sum_{n_1=0}^{m_1} \sum_{n_2=0}^{m_2}  q_1(n_1)q_2(n_2)V_t(m_1+\mathbf{1}_{\{i=1\}} - n_1, m_2+\mathbf{1}_{\{i=2\}} - n_2)
\end{equation}
\fi

\begin{description}
\item[Step 0:] Initialization: $V_0(m) = 0$ for all states $m=(m_1,\ldots, m_K)\in \mathcal{S}$, where $\mathcal{S}$ is state space. 
\item[Step 1:] Evaluate $V_{t+1}(m)~\forall~m \in \mathcal{S}$ using Bellmans optimality equation (\ref{optimality_eqn}).
\item[Step 2:] \textit{Stopping criterion:} If $\max\limits_{m\in \mathcal{S}} \{V_{t+1}(m) - V_{t}(m)\}- \min\limits_{m\in \mathcal{S}} \{V_{t+1}(m) - V_{t}(m)\} \le \epsilon,$ stop.   Otherwise increase $t$ by $t+1$ and go to step 1. 
\end{description}
\subsection{RSA, JSQ and JSEW}\label{sec:dispatch}
%In this section, we briefly provide the details of load balancing policies studied in the literature. We will compare the performance of our index policy with these dispatching rules in the numerical section.% and numerically illustrate that the index policy outperforms these dispatching rules. 
\begin{itemize}
\item \textbf{Random server allocation (RSA):} 
%Random server allocation is one of the easiest blind dispatching rule for server allocation. 
An incoming job is dispatched to a server chosen according to a uniform distribution. Information regarding the number of jobs in the server,  server scheduling disciplines, service rates, and cost of service is not taken into account while making a decision to choose a server. 

\item \textbf{Join the shortest queue (JSQ):} 
%This routing policy is widely studied and is applied  in server farms~\cite{gupta2007analysis}. 
An incoming job is routed to a server with
the least number of jobs, $n_k$, and ties are broken randomly. %This is an index policy, with index value $n_k$ for server~$k$. 
Thus, JSQ strives to balance load across the servers, reducing the probability of one server having several jobs while another server sits idle. JSQ is blind in terms of using the information of server speed and scheduling discipline.   

\item \textbf{Join the shortest expected workload (JSEW)} An incoming job is dispatched to the queue with the least expected workload, i.e.,  $n_k/q_k$ for server~$k$.
  This dispatching rule uses the information of server speed but is independent of the scheduling discipline of the server. %{\color{red}Thus, JSEW dispatches the jobs to the server with the highest index $q_k/n_k$.}

\end{itemize}

\section[Characteristics and Performance of index policy ]{Characteristics and Performance of index policy\footnote{All numerical codes are available in github repository: \url{https://github.com/manugupta-or/Limited_processor_sharing} } }\label{numerics}
%In this section, we discuss the characteristics of the index policy and compare the performance of our index policy with that of several other dispatching rules. 
%An important observation we make is that the Whittle index policy is very close to optimal. In addition, depending on the cost function chosen, the service disciplines employed in the servers can have a big impact on the optimal action. 

We consider a linear cost criterion in Section \ref{perf_linear_cost} and a weighted second order throughput based criterion (motivated by \cite{singh2015index}) capturing the mean variance trade-off in Section \ref{2nd_order_throughput}. Our main observations are that Whittle index policy is close to optimal under a wide range of parameter settings, and that with the weighted second order throughput criterion, the qualitative properties of Whittle's index policy strongly depend on the value of $d$ used in servers.

\subsection{Linear cost criterion}\label{perf_linear_cost}
In this section we assume a linear cost structure, i.e, $C_k(n)=n$ for all $k$.

\subsubsection{Characteristics of index policy}
We numerically examine the pattern of indices as a function of states. We consider three servers each with different scheduling policies ($d=\infty$ (PS), LPS-$d$, $d=1$ (FCFS)). An instance of indices patterns is shown in Figure \ref{pattern_indices} which confirms the non-increasing nature of Equation (\ref{general_index_value}) in Proposition \ref{proposition:monotone_index}.
\begin{figure}[!htbp]
  \centering
  \begin{minipage}[b]{0.49\textwidth}
\includegraphics[scale=0.47]{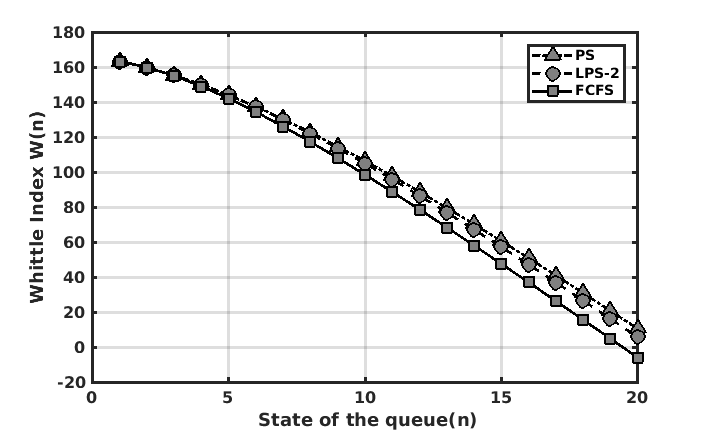}
\caption{Whittle's index for PS, LPS-2 and FCFS.}\label{pattern_indices}  \end{minipage}
  \begin{minipage}[b]{0.48\textwidth}
\includegraphics[scale=0.5]{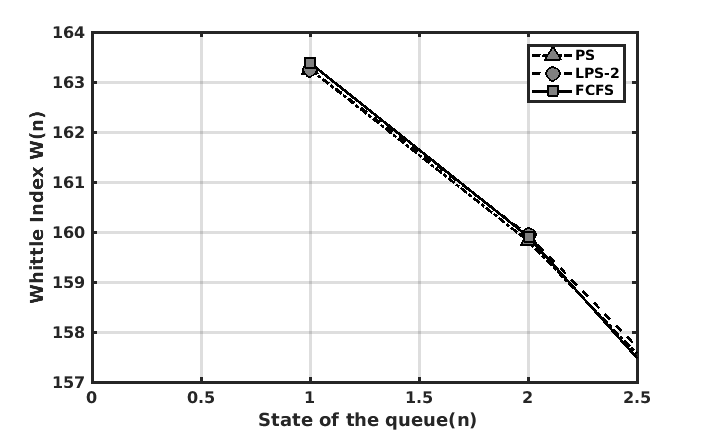}    
\caption{Highest index value with FCFS scheduling}\label{FCFS_fig}
  \end{minipage}
\end{figure}
\begin{figure}[!htbp]
  \centering
  \begin{minipage}[b]{0.48\textwidth}
\includegraphics[scale=0.5]{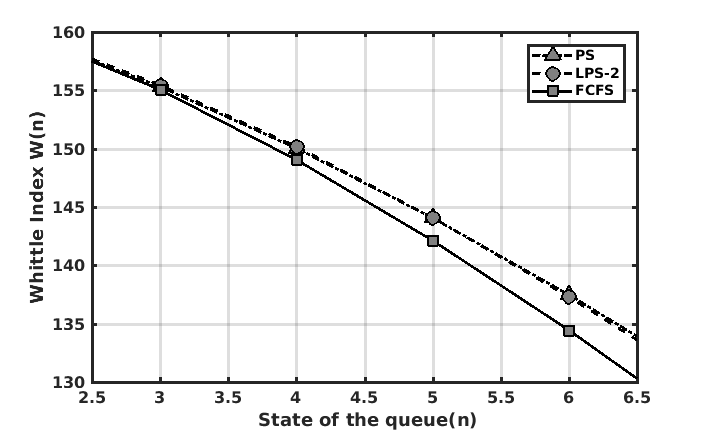}
\caption{Highest index value with LPS-$2$ scheduling}\label{LPS_fig}
 \end{minipage}
  \begin{minipage}[b]{0.48\textwidth}
\includegraphics[scale=0.5]{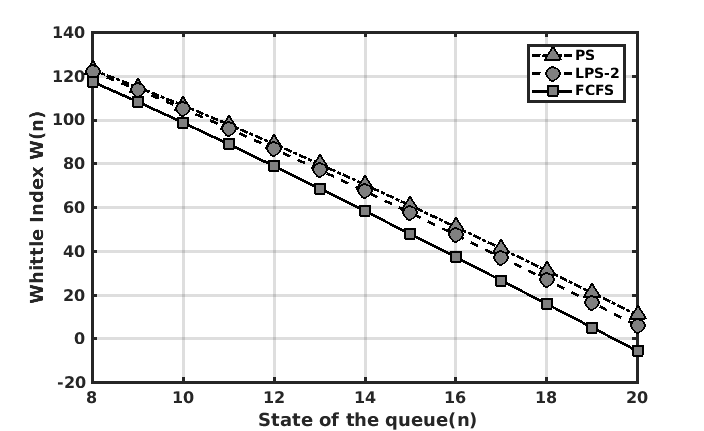}
\caption{Highest index value with PS scheduling}\label{PS_fig}
  \end{minipage}
\end{figure}
% which confirms the non-increasing nature of indices  as required in Proposition \ref{proposition:monotone_index}.    
We chose $D=300,~p=0.55,~q_1 =q_2=q_3 = 0.6,~d_1 = \infty ~(PS),~d_2=2~(LPS-2)\text{ and }d_3 = 1~(FCFS)$. We took equal  capacity for the three servers in order to assess the importance of scheduling disciplines employed at different servers when dispatching jobs. It can be observed that the service discipline has a small impact on the Whittle index. We magnify the index pattern of Figure \ref{pattern_indices} and present the instances where FCFS, LPS-$2$ and PS have the highest Whittle's index in Figures \ref{FCFS_fig}, \ref{LPS_fig}, and \ref{PS_fig}, respectively. It is evident from these figures that FCFS is preferred for small queues ($n\le 2$), LPS-2 is preferred for medium size queues ($2<n<7$) and PS is preferred for large queue sizes ($n\ge 7$). This phenomenon is intuitive: when there are not many jobs in the server, a new arriving job does not mind waiting, and receiving  the full service once in service (as under FCFS). On the other hand, jobs will prefer to receive the service immediately upon arrival when the queue length is large (as under PS). Further, note from Figure \ref{PS_fig} that PS queues with $n=13$ (higher Whittle's index) is preferred over the FCFS queue with $n=12$. Thus, the index policy accounts for the trade-off between waiting for the job's turn and getting the full dedication of the server against sharing the server but relatively less waiting in queue. 

%queue discipline and dispatches the job to appropriate server. 

\subsubsection{Comparison with JSEW/JSQ}
We consider two heterogeneous servers with service discipline parameters $d_1$ and $d_2$, respectively. 
 % Note that it was remarked in \cite[section 7]{gupta2007analysis} that it is unlikely that there is a routing policy which outperforms JSQ by more than about 10\% for the load balancing problem with processor sharing queues.
  % {\color{red} and upto ADD PS scheduling improvement}

% \begin{figure}[htb!]\centering
%\includegraphics[scale=0.5]{../../PERFORMANCE/plots/plots/scheduling_pattern.png}
%\caption{Job dispatching with heterogeneous speed and scheduler. Parameter settings: $p=0.5,~q_1 = 0.2,~q_2 = 0.7,~d_1 = 3,~d_2=5,~D=100$.  }\label{switching_dsdd}
%\end{figure}
\begin{figure}[!htbp]
  \centering
  \begin{minipage}[b]{0.48\textwidth}
\includegraphics[scale=0.7]{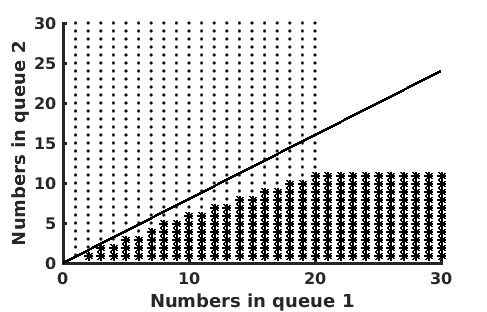}
\caption{Job dispatching with heterogeneous server speeds but same scheduling disciplines. Parameter settings: $p=0.3,~q_1 = 0.5,~q_2 = 0.4,~d_1 = 2,~d_2=2,~D=100$.}\label{switching_dssd}
 \end{minipage}\hfill
  \begin{minipage}[b]{0.48\textwidth}
\includegraphics[scale=0.7]{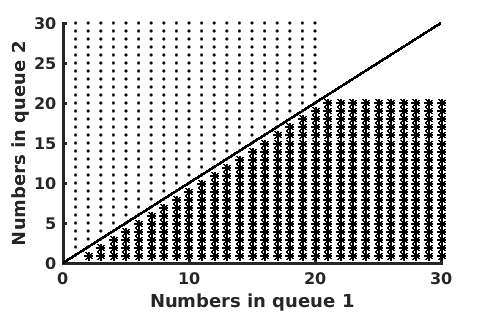}
\caption{Job dispatching with heterogeneous  scheduling disciplines but same server speeds. Parameter settings: $p=0.4,~q_1 = 0.6,~q_2 = 0.6,~d_1 = 1,~d_2=10,~D=100$.}\label{switching_ssdd}
  \end{minipage}
\end{figure}
%
%\begin{figure}[htb!]\centering
%\includegraphics[scale=0.5]{same_qs_different_ds.png}
%\caption{Job dispatching with heterogeneous  scheduling discipline but same server speed. Parameter settings: $p=0.55,~q_1 = 0.3,~q_2 =0.3,~d_1 = 1,~d_2=20,~D=100$. The total expected number with JSEW is 7.9366 and that with index policy is 7.9109 and the optimal policy gives 7.5493. For the quadratic cost, $C(i) = i^2$, the index policy restults in the cost 114.7684 and JSEW gives 115.2348. }\label{switching_ssdd}
%\end{figure}
In  Figures \ref{switching_dssd} and \ref{switching_ssdd}, we present the actions taken under Whittle's index policy. Figure \ref{switching_dssd} (Figure \ref{switching_ssdd}) is produced while varying the server speeds (service-disciplines) and keeping the same service disciplines (server speeds) in both servers. Parameter settings are as mentioned in the caption. 
In the area with ``$\cdot$" (``$\ast$"), queue 1 (queue 2) is prioritized in dispatching the jobs and in the white area jobs are blocked. 
In each of these figures, the straight line represents the switching curve for JSEW. 
Jobs are dispatched to queue 1 (queue 2) for the states above (below) the switching curve under JSEW. 

From the Figures \ref{switching_dssd} and \ref{switching_ssdd}, we see that Whittle's index policy differs from JSEW when $q_1\neq q_2$. 
For instance, having $q_1>q_2$ makes that Whittle's index policy dispatches more jobs to queue 1, see  Figure \ref{switching_dssd}.
On the other hand, from Figure~\ref{switching_ssdd} we observe that the scheduling discipline has hardly any impact on the load balancing decision: the actions taken under Whittle's index policy coincide with those under JSEW when $d_1=1$ and $d_2=10$. 
In general, we concluded from extensive numerical experiments, that 
the impact of $d$ on the Whittle index policy is negligible for linear cost functions. This is not surprising as it is well known that in the continuous time setting,  the evolution of the Markov process   will be independent on how the capacity is shared among the jobs, that is, is independent on the value of $d$.

%
%\begin{figure}[htb!]\centering
%\includegraphics[scale=0.5]{EnVsd.png}
%\caption{Job dispatching with heterogeneous  scheduling discipline but same server speed. Parameter settings: $p=0.75,~q_1 = 0.4,~q_2 =0.4,~d_2=1,~D=100$. The best performance is achieved at $d_1 = 1$. This is the general pattern noted in many configurations and the best performance is achieved at $d_1=d_2=1$. {\color{red} Can we prove this? or should we conjecture this?} }\label{change_with_d}
%\end{figure}

\subsubsection{Performance comparison}
In this section we compare the performance under our Whittle index policy with that of JSQ, JSEW and RSA.  We do not allow any blocking in the system.

\begin{figure}[!htbp]
  \centering
  \begin{minipage}[b]{0.48\textwidth}
\includegraphics[scale=0.76]{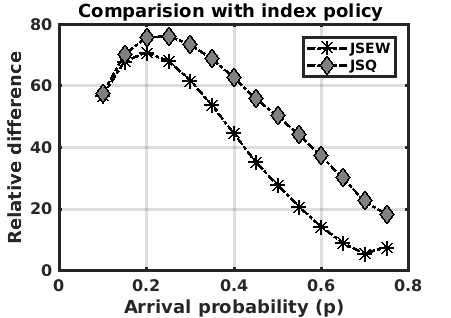}
\caption{Percentage relative difference in performance as compared to  Whittle's index policy. Parameter settings:  $q_1 = 0.1,~q_2 = 0.7,~d_1 = 3$ and $d_2 = 5$.
}\label{Compare_RSA}
 \end{minipage}\hfill
  \begin{minipage}[b]{0.48\textwidth}
\includegraphics[scale=0.4]{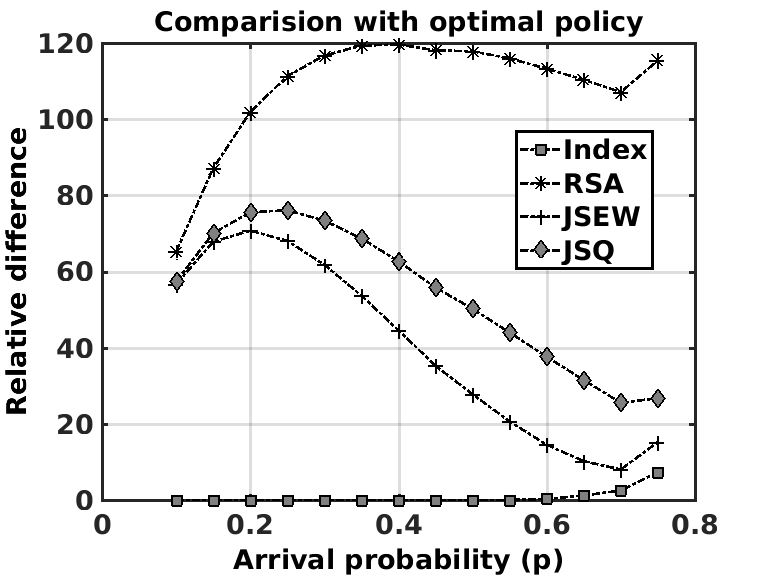}
\caption{ Percentage relative difference in performance as compared to the optimal policy. Parameter settings:  $q_1 = 0.1,~q_2 = 0.7,~d_1 = 3$ and $d_2 = 5$. }\label{Compare_JSQ2}
  \end{minipage}
\end{figure}

In Figure~\ref{Compare_RSA} and~\ref{Compare_JSQ2}  we plot the relative difference of the performance compared to that under Whittle's index policy  and that under an optimal policy, respectively, and let $p$ run from 0 till $q_1+q_2$ (the stability region). The relative difference (\%) (compared to the index policy) is computed as follows:
\begin{equation}
\text{Relative difference} = \frac{\mathbb{E}(N^{\psi})-\mathbb{E}(N^{Index})}{\mathbb{E}(N^{Index})}\times 100,
\end{equation}
where $\mathbb{E}(N^{\psi})$ is the expected number of jobs under policy $\psi \in \{\text{JSQ, JSEW}\}$. Similarly, the  relative difference (\%) (compared to the optimal policy) is computed as follows:
\begin{equation}\label{relative_imp}
\text{Relative difference} = \frac{\mathbb{E}(N^{\psi})-\mathbb{E}(N^{Opt})}{\mathbb{E}(N^{Opt})}\times 100,
\end{equation}
where $\mathbb{E}(N^{\psi})$ is the expected number of jobs under policy $\psi \in \{\text{RSA, JSQ, JSEW, Index}\}$. We implement the value iteration algorithm as described in Section \ref{optimal_policy} to compute the optimal performance. First of all, we notice that Whittle's index policy provides a stable system. 
Second, over a wide set of parameters, we notice the following general pattern. The Whittle index policy \textit{outperforms} all the standard dispatching rules and the random server allocation has by far the worst performance for higher values of $p$. 
% See Figure \ref{Compare_RSA} for reference. In this figure, we plot the performance of different dispatching rules while varying the arrival probability (equivalently, varying load factor). 

%
%\begin{figure}[htb!]\centering
%\includegraphics[scale=0.5]{plots/plots/JSQvsIndex2.png}
%\caption{Absolute improvement in performance for parameter setting in Figure \ref{Compare_JSQ}.}\label{absolute_imp}
%\end{figure}
%{\tt when figures changed, i can merge the below text with above text.}
%Since the performance of RSA is quite far from JSQ and JSEW, we drop RSA and compute the relative improvement in the performance of index policy with JSQ and JSEW. In Figure \ref{Compare_JSQ2}, we plot the performance for JSEW, JSQ and index policy for a system where server speeds are heterogeneous. We notice that index policy uniformly outperforms JSQ and JSEW across a wide set of load factors. We compute the percentage relative improvement with respect to JSQ/JSEW:
From Figure~\ref{Compare_RSA}, we notice a significant relative difference (upto 75\%) in performance with respect to JSQ/JSEW. Thus, one of the important conclusions of these numerical experiments is that the Whittle index policy \textit{significantly} outperforms all these standard dispatching rules. We further note from Figure \ref{Compare_JSQ2} that relative difference of Whittle's index policy compared to the optimal policy is very small (less than 3\%). We note that in other related load balancing problems, it has been established that Whittle's index policy is asymptotically optimal in light traffic, see for example~\cite{glazebrook2009index}.

%\subsubsection*{Comparison with the optimal policy}
%We also compare the performance of the Whittle's index policy with optimal scheduling scheme. We implement the value iteration algorithm as described in Section \ref{optimal_policy}. 

%
%
%\begin{figure}[!htbp]
%  \centering
%  \begin{minipage}[b]{0.48\textwidth}
%\includegraphics[scale=0.5]{optimalPerf_vs_index.png}
%\caption{Comparison of optimal performance vs index policy. Parameter settings:  $q_1 = 0.1,~q_2 = 0.7,~d_1 = 3$ and $d_2 = 5$.}\label{Compare_optimal}
% \end{minipage}\hfill
%  \begin{minipage}[b]{0.48\textwidth}
%\includegraphics[scale=0.5]{Relative_imp_optimalPerf_vs_index.png}
%\caption{Percentage relative loss in performance for parameter setting in Figure \ref{Compare_optimal}.}\label{Compare_optimal2}
%  \end{minipage}
%\end{figure}

%\begin{figure}[htb!]\centering
%\includegraphics[scale=0.5]{optimalPerf_vs_index.png}
%\caption{Comparision of optimal performance vs index policy. Parameter settings:  $q_1 = 0.1,~q_2 = 0.7,~d_1 = 3$ and $d_2 = 5$.}\label{Compare_optimal}
%\end{figure}
%\begin{figure}[htb!]\centering
%\includegraphics[scale=0.5]{Relative_imp_optimalPerf_vs_index.png}
%\caption{Percentage relative loss in performance for parameter setting in Figure \ref{Compare_optimal}.}\label{Compare_optimal2}
%\end{figure}

\subsection{Second order throughput cost criterion}\label{2nd_order_throughput}
In this section, we consider the problem of variance minimization that has been used in Markov decision processes to regularize systems (see \cite{kawai1987variance, kurano1987markov}). More specifically, we consider a cost criterion that consists of the linear cost and mean-variance trade off cost as in \cite{singh2015index}:
$$C_k(n) = \beta n + (1-\beta)\sum\limits_{i=1}^{\min\{n,{d_k}\}}(i^2-i\theta)q_k^{d_k}(i|n),$$
where $\theta$ is a parameter to tune the trade-off between mean throughput and the service regularity (second moment of throughput). By varying $\theta$, one can explore the entire Pareto frontier of mean-variance tradeoffs (see \cite{singh2015index}). And $\beta$ is the weight  associated with the linear cost term. Note that $\beta = 0$ would imply that the cost function, $C_k(n)$, is independent of $n$ (constant) for $n>d_k$. 

We can show that $q_k^{d_k}(i|n)$ is a non-decreasing function of $n$ (see \ref{non_decreasing_prob}). Thus,  $C_k(n)$ is a non-decreasing function in~$n$, hence, the Whittle index is as stated in~\eqref{general_index_value}, under the condition that $\sum_{m=0}^n \pi_k^n(m)$ is strictly increasing in~$n$.  

%We consider the weighted sum of mean number in system and above second order throughput as a performance criterion for the system to capture the mean-variance trade-off. We now comment on the characteristics and performance of index policy for such a performance measure. %We notice that the performance improves significantly with the index policy  as compare to JSQ/JSEW even for same server speed by just using the information on scheduling discipline. This is one of the major insights of the numerical experiments for such weighted cost criterion. \vspace{-0.2cm}

%
%\begin{figure}[!htbp]
%  \centering
%  \begin{minipage}[b]{0.48\textwidth}
%\includegraphics[scale=0.5]{state_space.png}
%\caption{Scheduling pattern.}\label{state_space_pattern_sameq}
% \end{minipage}\hfill
%  \begin{minipage}[b]{0.48\textwidth}
%\includegraphics[scale=0.5]{percentage_imp_Performance_JSWvsIndex.png}%\caption{Improvement in performance.}\label{improvement_same_d}
%  \end{minipage}
%\end{figure}

\subsubsection{Characteristics of index policy}
%We numerically notice that the indices are non-increasing for this cost criterion as well which confirms one of the requirement for Proposition \ref{proposition:monotone_index}. 

In Figure \ref{Indices_pattern_sameq} and \ref{Indices_pattern_diffq} we plot the Whittle index functions for an LPS-4 server and an LPS-6 server. We do this both for homogeneous server speeds (Figure~\ref{Indices_pattern_sameq}) and heterogeneous server speeds  (Figure~\ref{Indices_pattern_diffq}), respectively. We set $\beta = 0.001$ and $\theta = 0.9$.

\begin{figure}[!htbp]
  \centering
  \begin{minipage}[b]{0.48\textwidth}
\includegraphics[scale=0.65]{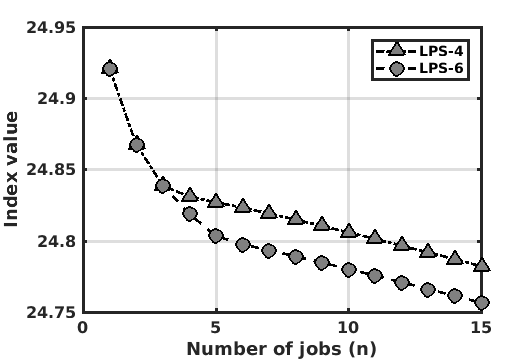}
\caption{Whittle's indices with homogeneous server speeds ($p=0.25,~q_1 = q_2 =0.3$).}\label{Indices_pattern_sameq}
 \end{minipage}\hfill
  \begin{minipage}[b]{0.48\textwidth}
\includegraphics[scale=0.65]{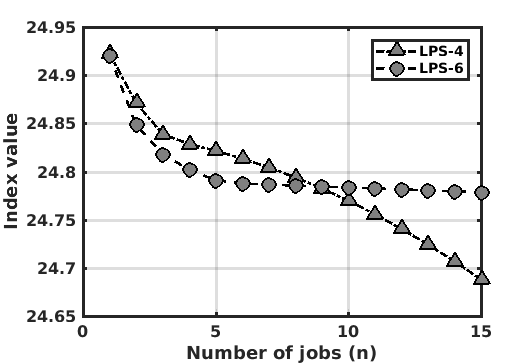}
\caption{Whittle's indices with heterogeneous server speeds ($p=0.25,~q_1 = 0.25, q_2 =0.5 $).}\label{Indices_pattern_diffq}
  \end{minipage}
\end{figure}

In Figure \ref{Indices_pattern_sameq}, it can be seen that when there are $n<4$ jobs in the system,  the index is the same for both servers. An explanation for this is that in such states, both servers equally share their capacity among all $n$ jobs. However, note that under LPS-6, there is a possibility that two more jobs enter and they will also receive a fair share of the capacity. Therefore, for $n<4$ we observe that both servers are equally attractive, however, when $n=4$,  LPS-4 is preferred over LPS-6, due to a possible new arrival of a job. 

We further observe in Figure \ref{Indices_pattern_sameq}  that $W^{LPS-4}(10)> W^{LPS-6}(5)$, that is, Whittle's index policy chooses an LPS-4 server with 15 (or less) jobs over an LPS-6 server with 5 (or more) jobs. This preference for a higher loaded server comes from the fact that the LPS-4 system has a smaller variation in terms of second-order throughput and the cost criterion consist of such term. Similar observations can be made  in Figure \ref{Indices_pattern_diffq} for heterogeneous server speeds.

%
%\begin{figure}[!htbp]
%  \centering
%  \begin{minipage}[b]{0.48\textwidth}
%\includegraphics[scale=0.6]{state_space1.png}
%\caption{Scheduling pattern in index policy for homogeneous server speeds  ($q_1 =q_2 = 0.3$).}\label{state_space_pattern_sameq}
% \end{minipage}\hfill
%  \begin{minipage}[b]{0.48\textwidth}
%\includegraphics[scale=0.4]{percentage_imp_Performance_JSWvsIndex.png}%\caption{Percentage relative loss in performance for homogeneous server speeds.}\label{improvement_same_d}
%  \end{minipage}
%\end{figure}

 %We also plot absolute weighted performance in Figure \ref{performance_comp} for illustration. 

\begin{figure}[htp]
\centering
\includegraphics[width=.33\textwidth]{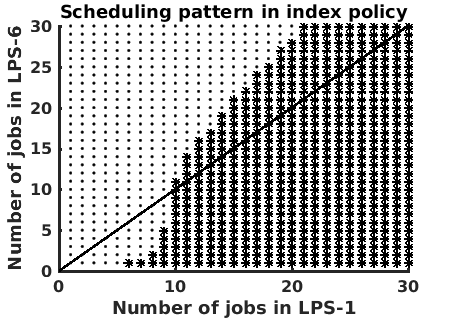}\hfill
\includegraphics[width=.33\textwidth]{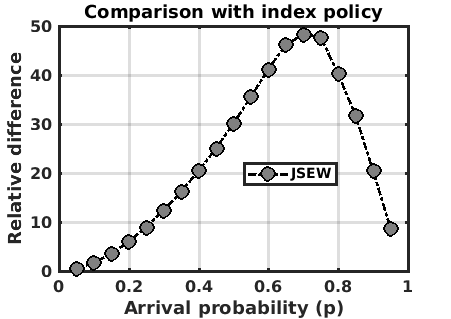}\hfill
\includegraphics[width=.33\textwidth]{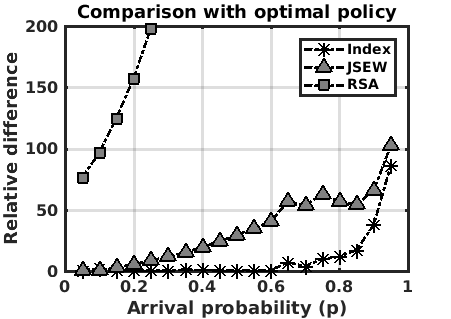}
\caption{Scheduling pattern (left) under index policy with $p=0.65$, percentage relative difference (middle and  {right}) as compare to index and optimal (respectively) policy for homogeneous server speeds  ($q_1 =q_2 = 0.5$) and $d_1 = 1, d_2 =6$. }
\label{homogeneous_queues}
\end{figure}

\subsubsection{Performance comparison}

We present the scheduling pattern in Figure \ref{homogeneous_queues} (left) for homogeneous server speeds with parameter settings mentioned in the caption. That is, the servers only differ in their implemented service policies. We observe that this difference has a large impact on the scheduling pattern:  Whittle's index policy might send jobs to LPS-1, even though this server is higher loaded than the LPS-6 server and vice-versa depending on the state of the system. This is unlike with linear costs, where we saw in Figure~\ref{switching_ssdd} that under Whittle's index policy one would send the job to the server with the least number of jobs, even though the servers have different implemented policies. We also plot the switching curve for JSQ and JSEW in Figure~\ref{homogeneous_queues} (left)  (straight line with $45^\circ$ slope). We observe that the scheduling pattern in Figure \ref{homogeneous_queues} is quite different from JSQ/JSEW. 
In Figure~\ref{homogeneous_queues} (middle and right), we plot the relative difference (in \%) when comparing to  Whittle's index policy and an optimal policy, respectively.  We notice that losses under the Whittle  index policy are uniformly minimum across different load factors as compare to other dispatching rules (JSQ, JSEW, RSA). 

We perform another set of experiments with heterogeneous server speeds for the same parameters setting as in Figure \ref{Indices_pattern_diffq} and notice a similar phenomenon (see Figure \ref{hetrogeneous_queues}). %Similar to the case of linear holding cost, performance of the Whittle's index policy has smaller losses as compare to the optimal policy for this weighted cost structure ({\color{red}ADD OPTIMAL FIGURE}).
%
%\begin{figure}[!htbp]
%  \centering
%  \begin{minipage}[b]{0.48\textwidth}
%\includegraphics[scale=0.6]{state_space_qcost_diff_server_speed.png}%{state_space_pattern_sameq}
%\caption{Scheduling pattern in index policy for heterogeneous server speeds  ($q_1 = 0.25$ and $q_2 = 0.5$).}\label{state_space_pattern_new_cost_diffq}
% \end{minipage}\hfill
%  \begin{minipage}[b]{0.48\textwidth}
%\includegraphics[scale=0.4]{percentage_imp_Performance_JSWvsIndex_diffq.png}
%\caption{Percentage relative loss in performance for heterogeneous server speeds.}\label{improvement_new_cost_diffq}
%  \end{minipage}
%\end{figure}
%   
%
%
\begin{figure}[htp]

\centering
\includegraphics[width=.33\textwidth]{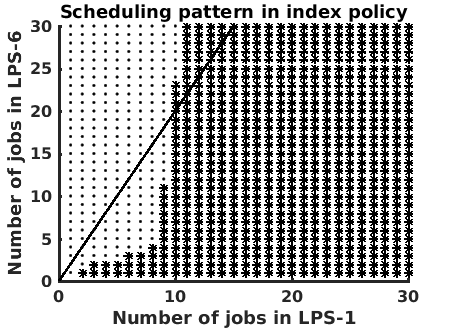}\hfill
\includegraphics[width=.33\textwidth]{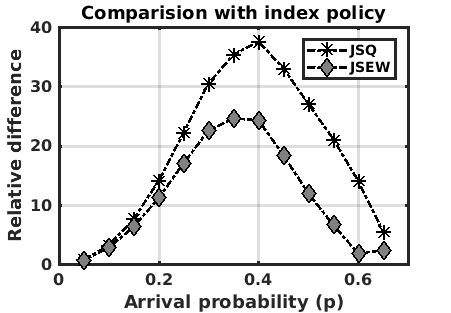}\hfill
\includegraphics[width=.33\textwidth]{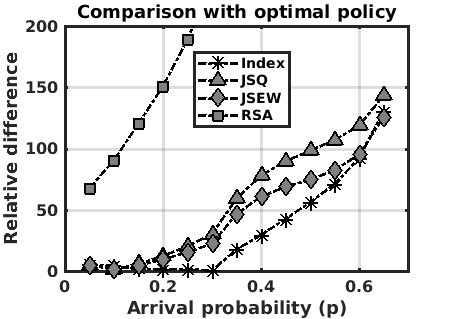}

\caption{Scheduling pattern (left) under index policy with $p=0.4
$, percentage relative difference (middle and {right}) as compare to index and optimal (respectively) policy for heterogeneous server speeds ($q_1 = 0.2$ and $q_2 = 0.5$) and $d_1 = 1, d_2 =6$. 
}
\label{hetrogeneous_queues}

\end{figure}

%
%\begin{figure}[htb!]\centering
%\includegraphics[scale=0.5]{Performance_JSWvsIndex.png}
%\caption{Performance comparison.}\label{performance_comp}
%\end{figure}

\section{Concluding remark}\label{sec:conclusion}
We have considered the load-balancing problem for LPS-$d$ type of schedulers with a possibility of blocking. Our solution approach is from a restless bandit perspective and one of the important take away messages of the study is that a well performing dispatcher can depend on the scheduling discipline deployed in the servers. 
%In our experiment, we notice that the sub-optimality gap is small. Thus, index policy takes different parameters into account and appropriately dispatches the jobs. 

Several possible extensions suggest themselves. The impact of abandonment and batch arrivals on the load-balancing problem can be an immediate follow up. The performance of Whittle's dispatching rule in a many-server regime, also known as mean-field regime, would be interesting, as it might yield closed-form expressions.  The stability of Whittle's index policy for load balancing problem is another interesting future avenue. Extending the results of this work to load balancing systems with multiple dispatchers, as it happens in modern data centers, is also relevant. In particular, this would require to capture the impact that one dispatcher's action has upon others. 

%There are many complicated models of load balancing and it would be interesting to see the performance of index policies for those load-balancing problems. 

\section*{Acknowledgement}
This research is partially supported by the French Agence Nationale de la Recherche (ANR) through the project ANR-15-CE25-0004 (ANR JCJC
RACON) and  by ANR-11-LABX-0040-CIMI within the program 
ANR-11-IDEX-0002-02.

%\begin{bibliography}
\bibliographystyle{plain}
%\bibliographystyle{plain}
%\bibliographystyle{model5-names}
%\bibliography{../../Operations-Research-template/references}%\end{bibliography}
\bibliography{references}%\end{bibliography}

\appendix
\section{Proof of propositions and lemma}\label{appendix:proofFCFS}
We present various proofs in this section. 
\subsection{Proof of Lemma \ref{Lemma:stoch_dom}}
{Some ideas in this proof are borrowed from \cite{borkar2017whittle}.} We drop the dependency on $k$ for ease of notations. We prove this result by using the idea of stochastic dominance in Markov Chains. Consider two Markov chains, $\{X_t^1\},~\{X_t^2\}$ corresponding to threshold $n$ and $n+1$. Label the positive recurrent states in the reverse order, i.e., if the threshold is $x$, state with $(x + 1)$ jobs
is labeled as state 0, that with $x$ jobs is labeled as state 1, and this is repeated till the state with 0 jobs
is labeled $(x + 1)$. We will show that $\{X_t^2\}$ stochastically dominates $\{X_t^1\}$ for LPS-$d$ scheduling discipline.

Let the transition matrices of the Markov chains be $P_1$ and $P_2$ for $\{X_t^1\}$ and $\{X_t^2\}$ respectively. Their elements $p_k(i, j)$, $
k = 1, 2$, give the corresponding probability of going from state $i$ (i.e., $(x + 1 - i)$ jobs in the server) to state $j$ (i.e., ($x + 1 - j$) jobs in the server). Extend $P_1$ to
a $(n + 3) \times (n + 3)$ matrix by adding a column and row of zeros:{
$$ 
 \begin{pmatrix}
p_{n+1, n+1} & p_{n+1, n} &p_{n+1, n-1} & \cdots  & p_{n+1, 1}& p_{n+1, 0} &0 \\
 p_{n, n+1} & p_{n, n} & p_{n, n-1} &  \cdots & p_{n, 1} & p_{n, 0} &0 \\0 &   p_{n-1,n} &  p_{n-1,n-1} & \cdots & p_{n-1, 1} & p_{n-1, 0} &0 \\
  \vdots  & \vdots  & \vdots  & \ddots &\vdots &\vdots&\vdots \\
0 & 0 & 0  &   \cdots & p_{1,1}  &p_{1,0} &0 \\ 
0 & 0 & 0  &  \cdots & p & 1-p & 0 \\
0 & 0 & 0  &  \cdots & 0 & 0 & 0 \\
 \end{pmatrix}
$$
and the matrix $P_2$ is:
$$ 
 \begin{pmatrix}
p_{n+2, n+2} & p_{n+2, n+1} &p_{n+2, n} & \cdots  & p_{n+2, 2}& p_{n+2, 1} & p_{n+2, 0} \\
 p_{n+1, n+2} & p_{n+1, n+1} & p_{n+1, n} &  \cdots & p_{n+1, 2} & p_{n+1, 1} &p_{n+1, 0} \\
0 &   p_{n,n+1} &  p_{n,n} & \cdots & p_{n, 2} & p_{n, 1} & p_{n, 0} \\
  \vdots  & \vdots  & \vdots  & \ddots &\vdots &\vdots&\vdots \\
0 & 0 & 0  &   \cdots & p_{2,2}  & p_{2,1} & p_{2,0} \\ 
0 & 0 & 0  &  \cdots & p_{1,2}  & p_{1,1}  & p_{1,0} \\
0 & 0 & 0  &  \cdots & 0 & p & 1-p \\
 \end{pmatrix}
$$}
where each element in above matrices is according to the Bernoulli arrival Process and a departure process driven by LPS-$d$ scheduling scheme (See section \ref{LPS-d} for the precise definition of $P_{i,j}$ in above matrices). Now consider the following lower triangular matrix of dimension $(n + 2) \times (n + 2)$:
$$ 
U =  \begin{pmatrix}
1 & 0 &0 & \cdots &0 \\
  1 & 1 &0 &  \cdots &0 \\
1 &   1 &  1 &\cdots &0 \\
  \vdots  & \vdots  & \vdots   &\ddots &\vdots \\
1 &   1 &  1 &\cdots & 1\\
 \end{pmatrix}
$$
It follows from above matrices that
$$P_1U \le P_2 U.$$
This shows that $\{X_t^2\}$ stochastically dominates $\{X_t^1\}$ (see \cite[Definition 3.12]{kijima2013markov}). Further, using Theorem 3.31 in \cite[Page 158]{kijima2013markov}, we have 
$$\bar{\pi}^{n+1}(0) \le \bar{\pi}^{n}(0)$$
where $\bar{\pi}^n$ (resp., $\bar{\pi}^{n+1}$) is the stationary distribution for the threshold $n$ (resp., $n + 1$) with relabeled states. In view of our relabeling of the positive recurrent states, this
translates into the following in the original notation:
$${\pi}^{n+1}(n+2) \le {\pi}^{n}(n+1)$$
This shows that the stationary probability of the only passive state with positive stationary probability decreases as the threshold increases. However,
$$\sum_{j=0}^{n+1}\pi^n(j) =1$$
and $\pi^n(n+1)$ decreases with the threshold $n$. This shows that
$\sum\limits_{m=0}^n\pi^n(m)$ is non-decreasing in $n$. 
\subsection{Proof of Proposition \ref{prop:FCFS}}
We drop the subscript $k$ for ease of notation. 
The Whittle's index is given by Equation (\ref{general_index_value}) if it is non-increasing in $n$:
$$pD-\frac{ \sum\limits_{m=0}^{n}C(m)[\pi^{n}(m)-\pi^{n-1}(m)] + C(n+1)\pi^{n}(n+1) }{\sum\limits_{m=0}^n\pi^n(m) - \sum\limits_{m=0}^{n-1}\pi^{n-1}(m)}$$
Above index can be rewritten as: 
$$pD-\frac{ \sum\limits_{m=0}^{n-1}C(m)[\pi^{n}(m)-\pi^{n-1}(m)] + C(n)[\pi^{n}(n)-\pi^{n-1}(n)] + C(n+1)\pi^{n}(n+1) }{\sum\limits_{m=0}^n\pi^n(m) - \sum\limits_{m=0}^{n-1}\pi^{n-1}(m)}$$
The sum of stationary probability simplifies to the following by using Equations (\ref{stat_prob_0}) - (\ref{stat_prob_n}):
$$\sum_{m=0}^n\pi^n(m) = \frac{q^{n+2}(1-p)^{n} - p^{n+1}q(1-q)^n}{q^{n+2}(1-p)^n - p^{n+2}(1-q)^n}, $$
and the term, $\sum\limits_{m=0}^n\pi^n(m) - \sum\limits_{m=0}^{n-1}\pi^{n-1}(m)$,  simplifies to 
$$ \frac{(1-p)^{n-1}(1-q)^{n-1}p^nq^{n+1}(p-q)^2}{(q^{n+2}(1-p)^n - p^{n+2}(1-q)^n)(q^{n+1}(1-p)^{n-1} - p^{n+1}(1-q)^{n-1})},$$
by using the stationary distribution for FCFS discipline (derived in Section \ref{FCFS_scheme}). We have  
$$\frac{\pi^{n}(m)-\pi^{n-1}(m)}{\sum\limits_{m=0}^n\pi^n(m) - \sum\limits_{m=0}^{n-1}\pi^{n-1}(m)} = -\frac{p^2\left(\frac{b}{d}\right)^{m-1}}{(1-p)q^2}~\forall~m=1,2,\cdots,n-1,$$
$$\frac{\pi^{n}(n) - \pi^{n-1}(n)}{\sum\limits_{m=0}^n\pi^n(m) - \sum\limits_{m=0}^{n-1}\pi^{n-1}(m)} = \frac{q}{q-p}\left[ p + \left(\frac{p}{q}\right)^{n+1}\left(\frac{1-q}{1-p}\right)^{n-1} \left(\frac{p}{q}-p-1 \right)\right],$$
and
$$\frac{\pi^{n}(n+1)}{\sum\limits_{m=0}^n\pi^n(m) - \sum\limits_{m=0}^{n-1}\pi^{n-1}(m)} = \frac{p(1-q)}{q-p}\left[ 1-\left(\frac{p}{q}\right)^{n+1}\left(\frac{1-q}{1-p}\right)^{n-1}\right]$$

Using above expressions, the Whittle's index simplifies to the expression as stated in the proposition. Note that we are left to argue the non-increasing nature of the index. We argue this in the following steps:
After sum algebra we can write 
%$$W_k(n) = pD+ \frac{p^2}{q_k^2(1-p)}\sum_{m=0}^{n-1}C_k(m)\left(\frac{b_k}{d_k}\right)^{m-1} -\frac{C_k(n)q_k}{(q_k-p)}\left[p + \left(\frac{p}{q_k}\right)^{n+1}\left(\frac{1-q_k}{1-p}\right)^{n-1}\left(\frac{p}{q_k}-p-1 \right) \right] $$
%\begin{equation}\label{Index_FCFS}  -\frac{C_k(n+1)p(1-q_k)}{(q_k-p)}\left[1- \left(\frac{p}{q_k}\right)^{n+1}\left(\frac{1-q_k}{1-p}\right)^{n-1} \right], \end{equation}
%$$W_k(n+1) = pD+ \frac{p^2}{q_k^2(1-p)}\sum_{m=0}^{n}C_k(m)\left(\frac{b_k}{d_k}\right)^{m-1} -\frac{C_k(n+1)q_k}{(q_k-p)}\left[p + \left(\frac{p}{q_k}\right)^{n+2}\left(\frac{1-q_k}{1-p}\right)^{n}\left(\frac{p}{q_k}-p-1 \right) \right] $$\begin{equation}\label{Index_FCFS}  -\frac{C_k(n+2)p(1-q_k)}{(q_k-p)}\left[1- \left(\frac{p}{q_k}\right)^{n+2}\left(\frac{1-q_k}{1-p}\right)^{n} \right], \end{equation}
$$W_k(n+1) - W_k(n) = \frac{C_k(n)pq_k}{q_k - p}\left[1 - \left(\frac{p}{q}\right)^2\left(\frac{p(1-q_k)}{q_k(1-p)}\right)^n\right] + \frac{C_k(n+1)p(1-2q_k)}{q_k - p}\left[1 - \left(\frac{p}{q}\right)^2\left(\frac{p(1-q_k)}{q_k(1-p)}\right)^n\right]$$
$$  -\frac{C_k(n+2)p(1-q_k)}{(q_k-p)}\left[1- \left(\frac{p}{q_k}\right)^{n+2}\left(\frac{1-q_k}{1-p}\right)^{n} \right].
$$
%On further re-writing, we have:
%$$W_k(n+1) - W_k(n) =\left( \frac{C_k(n)pq_k}{q_k - p} + \frac{C_k(n+1)p(1-2q_k)}{q_k - p} -\frac{C_k(n+2)p(1-q_k)}{(q_k-p)}\right) \left[1 - \left(\frac{p}{q}\right)^2\left(\frac{p(1-q_k)}{q_k(1-p)}\right)^n\right], $$
Further simplification results in 
\begin{equation}\label{index_difference}
W_k(n+1) - W_k(n) =\left( \frac{pq_k (C_k(n)-C_k(n+1)) + p(1-q_k)(C_k(n+1)-C_k(n+2))}{q_k - p}\right) \left[1 - \left(\frac{p}{q_k}\right)^2\left(\frac{p(1-q_k)}{q_k(1-p)}\right)^n\right]. 
\end{equation}
Note that the term $$\frac{1}{q_k - p}\left[1 - \left(\frac{p}{q_k}\right)^2\left(\frac{p(1-q_k)}{q_k(1-p)}\right)^n\right] >0. $$
Thus, the sign of the term $W(n+1) - W(n)$ is decided by  $pq_k (C_k(n)-C_k(n+1)) + p(1-q_k)(C_k(n+1)-C_k(n+2))$. Note that this term is negative from the assumption that the cost function is non-decreasing. Thus, $W(n+1) - W(n) \leq 0,~\forall~n$. This completes the proof for the non-increasing nature of the index.  

We now give a proof for the second part of the proposition. It immediately follows with some simplification from the expression of the index that  
$$W_k(0) = pD + C_k(1)p\frac{q_k+p-1}{q_k},$$

We consider the following two cases to determine $\lim\limits_{n\rightarrow\infty}W_k(n)$.
\begin{description}
\item[Case $\mathbf{q_k<p}$:]{
 It follows from Equation (\ref{index_difference}) and the fact that the cost function is non-decreasing, that} 
$\lim\limits_{n\rightarrow\infty}W_k(n+1) - W_k(n) \rightarrow -\infty$. Since $W_k(n)$ is non-increasing, we then have $\lim\limits_{n\rightarrow\infty}W_k(n)\rightarrow-\infty$. 
\item[Case $\mathbf{q_k>p}$:]
From~\eqref{Index_FCFS}, we have that 
\begin{equation}\label{Index_FCFS}
\lim\limits_{n\rightarrow\infty}W_k(n)=    pD + \frac{p^2}{q_k^2(1-p)}\sum_{m=0}^{\infty}C_k(m)\left(\frac{b_k}{d_k}\right)^{m-1} - \frac{\lim\limits_{n\to \infty}C_k(n)q_k p }{(q_k-p)} - \frac{\lim\limits_{n\to \infty}C_k(n+1)(1-q_k)p }{(q_k-p)}.
\end{equation}
Since $b<d$,  $C_k(n)$ is bounded by a polynomial and $\lim\limits_{n\to \infty}C_k(n)\to \infty$, the above limit converges to $-\infty$. This completes the proof.  
\end{description}

\subsection{Proof of Equation~\eqref{linear_cost_FCFS_index}}

 Consider the linear cost, i.e., $C_k(m) = mC~\forall~m$. Using Equation (\ref{Index_FCFS}), the index term simplifies as stated below:
$$pD+ \frac{Cp^2 (1-p)}{(q_k-p)^2}\left[ 1-n\left(\frac{b}{d}\right)^{n-1} + (n-1)\left(\frac{b}{d}\right)^{n}  \right] -\frac{nCq_k}{(q_k-p)}\left[p + \left(\frac{p}{q_k}\right)^{n+1}\left(\frac{1-q_k}{1-p}\right)^{n-1}\left(\frac{p}{q_k}-p-1 \right) \right] $$
\begin{equation}\nonumber
  -\frac{(n+1)Cp(1-q_k)}{(q_k-p)}\left[1- \left(\frac{p}{q_k}\right)^{n+1}\left(\frac{1-q_k}{1-p}\right)^{n-1} \right]. 
\end{equation}
Above, further simplifies to: 
$$pD+ \frac{Cp^2 (1-p)}{(q_k-p)^2}\left[ 1-n\left(\frac{b_k}{d_k}\right)^{n-1} + (n-1)\left(\frac{b_k}{d_k}\right)^{n}  \right]$$
$$-\frac{Cp(1-q_k)}{q_k-p} - \frac{nCp}{q_k-p} + \frac{Cp(1-q_k)}{q_k-p}\left(\frac{p}{q_k}\right)^{2} \left(\frac{b_k}{d_k}\right)^{n-1} + \frac{nCq_k}{q_k-p}\left(\frac{p}{q_k}\right)^{2} \left(\frac{b_k}{d_k}\right)^{n-1}. $$
On rearranging the terms, above can be written as: 
\begin{equation}\nonumber
 pD+ \frac{Cp^2 (1-p)}{(q_k-p)^2} -\frac{Cp(1-q_k)}{q_k-p} - \frac{nCp}{q_k-p}- \frac{Cp^3(1-p)}{q_k(q_k-p)^2}\left(\frac{b_k}{d_k}\right)^{n}.
\end{equation}
It can be easily verified that above term is non-increasing in $n$ in either of the cases $q<p$ or $q>p$. Thus, the Whittle's index  
\begin{equation}\nonumber
 W_k(n) = pD+ \frac{Cp^2 (1-p)}{(q_k-p)^2} -\frac{Cp(1-q_k)}{q_k-p} - \frac{nCp}{q_k-p}- \frac{Cp^3(1-p)}{q_k(q_k-p)^2}\left(\frac{b}{d}\right)^{n}.
\end{equation}
Further, it immediately follows from above expression at $n=0$ that 
$$ W_k(0) = pD+ \frac{Cp^2 (1-p)}{(q_k-p)^2} -\frac{Cp(1-q_k)}{q_k-p}- \frac{Cp^3(1-p)}{q_k(q_k-p)^2},$$
which simplifies to $W_k(0)$ stated in the proposition. 

\subsection{$q_k^{d_k}(i|n)$ is non-decreasing function of $n$.}\label{non_decreasing_prob}
By definition,
$$ q_k^{d_k}(i|n) = {\min\{n,d_k\} \choose i} \left(\frac{q_k}{\min\{n,d_k\}}\right)^i\left(1-\frac{q_k}{\min\{n,d_k\}}\right)^{\min\{n,d_k\}-i},
$$
note that it is straight forward to see that above function is constant in $n$ for $n>d_k$. We prove below non-decreasing nature of $q_k^{d_k}(i|n)$ for $n\le d_k$. For $n \le d_k$, 
$$ q_k^{d_k}(i|n) = {n \choose i} \left(\frac{q_k}{n}\right)^i\left(1-\frac{q_k}{n}\right)^{n-i}.
$$
Above further simplifies to
$$ q_k^{d_k}(i|n) = \frac{q^x}{(i)!} \frac{n(n-1)(n-2)\dots(n-i+1)}{n^i} \left(1-\frac{q_k}{n}\right)^{n-i}.
$$
It can be easily argued that the terms $\frac{n(n-1)(n-2)\dots(n-i+1)}{n^i}$ and $\left(1-\frac{q_k}{n}\right)^{n-i}$ are positive and increasing in $n$. Hence, $q_k^{d_k}(i|n)$ is increasing in $n$.

\end{document}